\documentclass{article}
\usepackage{fullpage}
\usepackage{setspace}
\usepackage[margin=1in]{geometry}
\usepackage{amsmath,amssymb,amsthm,mathtools}
\usepackage{booktabs,tabularx,array}
\usepackage{enumitem}
\usepackage{microtype}
\usepackage{placeins}
\usepackage{xcolor}
\usepackage[mathlines]{lineno}

\usepackage{tikz}
\usetikzlibrary{arrows.meta,positioning}
\tikzset{vf/.style={circle,draw=darkblue,minimum size=6mm,inner sep=1pt,font=\small}, ve/.style={-{Latex[length=2mm]},thick,draw=darkblue}, vd/.style={ve,dashed}, note/.style={font=\small,align=center}, boxnote/.style={draw=darkblue!40,fill=softblue,rounded corners=2pt,inner sep=8pt,align=left,font=\small}}

\usepackage{float}
\usepackage{needspace}
\usepackage{url}
\usepackage[hidelinks]{hyperref}
\usepackage[nameinlink,noabbrev]{cleveref}

\newtheorem{theorem}{Theorem}[section]
\newtheorem{lemma}[theorem]{Lemma}
\newtheorem{proposition}[theorem]{Proposition}
\newtheorem{corollary}[theorem]{Corollary}

\theoremstyle{definition}
\newtheorem{definition}[theorem]{Definition}
\newtheorem{algorithmdef}[theorem]{Algorithm}
\theoremstyle{remark}

\newcommand{\Anc}{\operatorname{Anc}}
\newcommand{\Desc}{\operatorname{Desc}}
\newcommand{\rankv}{\operatorname{rank}}
\newcommand{\sizev}{\operatorname{size}}
\newcommand{\sid}{\operatorname{id}}

\newcommand{\gcl}{\operatorname{cl}_{\Gamma}}

\newcommand{\calF}{\mathcal F}

\newcommand{\bits}{\{0,1\}}

\definecolor{darkblue}{RGB}{25,65,115}
\definecolor{softblue}{RGB}{240,246,252}
\definecolor{darkgreen}{RGB}{27,102,60}
\definecolor{softgreen}{RGB}{239,249,242}

\setlist{topsep=3pt,itemsep=2pt,parsep=1pt}

\renewenvironment{abstract}
  {\begin{center}\bfseries\abstractname\end{center}\begin{quote}\normalsize}
  {\end{quote}}
\hypersetup{pdfauthor={Sankardeep Chakraborty and Seungbum Jo},
  pdftitle={Optimal Labeling Schemes for Gamma-Free Matrices}}
\title{Optimal Labeling Schemes for $\Gamma$-Free Matrices}
\author{%
Sankardeep Chakraborty\\[3pt]
\normalsize The University of Tokyo\\
\normalsize Tokyo, Japan\\
\small\href{mailto:sankardeep.chakraborty@gmail.com}{\texttt{sankardeep.chakraborty@gmail.com}}
\and
Seungbum Jo\\[3pt]
\normalsize Chungnam National University\\
\normalsize Daejeon, Republic of Korea\\
\small\href{mailto:sbjo@cnu.ac.kr}{\texttt{sbjo@cnu.ac.kr}}
}
\date{}

\begin{document}
\hypersetup{pageanchor=false}
\maketitle
\thispagestyle{empty}
\vspace{-2.2em}

\begin{abstract}
We study labeling schemes for $p \times q$ $\Gamma$-free binary matrices,
which do not contain a submatrix of the form
$\Gamma = \begin{pmatrix} 1 & 1 \\ 1 & 0 \end{pmatrix}$.
We give a labeling scheme that uses $O(\log^2 N)$-bit labels
for rows and columns, where $N=p+q$, such that any matrix entry can be determined from the
corresponding two labels in constant time in the word RAM
model.
The worst-case label size matches the counting lower bound
up to constant factors.
We also show that the entrywise least $\Gamma$-free
supermatrix of an arbitrary $p\times q$ binary matrix,
together with its labels, can be constructed in $O(pq)$ time.

Our construction represents a $\Gamma$-free matrix by a sequence
of rooted forests, where each column corresponds to a node and
the $1$-entries of each row form a root path in its corresponding
forest.

For $n$-vertex strongly chordal and chordal bipartite graphs,
our result gives asymptotically optimal $O(\log^2 n)$-bit
adjacency labeling schemes with constant-time queries.
These results answer both representation questions posed by
Spinrad in 1995.
\end{abstract}

\clearpage
\setcounter{page}{1}
\hypersetup{pageanchor=true}

\section{Introduction}
\label{sec:introduction}
A binary matrix $M$ is said to be \textit{$\Gamma$-free} if it does not contain a $2 \times 2$ submatrix of the form $\Gamma = \begin{pmatrix} 1 & 1 \\ 1 & 0 \end{pmatrix}$,
where the submatrix is obtained by selecting two rows and two columns
without changing their order.
$\Gamma$-free matrices are closely related to several well-studied
graph classes. In particular, a binary matrix is totally balanced if and only if
its rows and columns can be permuted independently to make it
$\Gamma$-free~\cite{AnsteeFarber84}. This characterization applies to the biadjacency matrices of chordal bipartite graphs~\cite{GolumbicGoss78,AnsteeFarber84}.
Moreover, a graph is strongly chordal if and only if its closed-neighborhood matrix can be made $\Gamma$-free by applying the same vertex ordering to its rows and columns
~\cite{Farber83,HellHernandezHuangLin21}. 

An \textit{adjacency labeling scheme} assigns a binary label to each vertex
of a graph such that the adjacency of two vertices can be determined
using only their labels~\cite{KannanNaorRudich92}. The number of labeled $n$-vertex graphs in a class gives an information-theoretic lower bound on the label size,
but a matching upper bound does not necessarily exist.
For example, Hatami and Hatami~\cite{HatamiHatami22}
showed that some hereditary graph classes (closed under taking induced subgraphs) with $2^{O(n\log n)}$ labeled graphs require $n^{\Omega(1)}$-bit labels.
More recently, Bonnet, Duron, Sylvester, Zamaraev, and Zhukovskii~\cite{BonnetDuronSylvesterZamaraevZhukovskii24} showed that some monotone graph classes (closed under taking subgraphs) with $2^{O(n\log^2 n)}$ labeled graphs require
$\Omega(\log^3 n)$-bit labels. We refer to \Cref{sec:prior} for a more detailed discussion
of previous results on adjacency labeling schemes.

Spinrad~\cite{Spinrad95} showed that the number of
$m\times m$ $\Gamma$-free matrices is
$2^{\Theta(m\log^2 m)}$ and gave an asymptotically
optimal space representation for these matrices.
However, this representation does not provide efficient access
to individual matrix entries, nor does it directly yield
an adjacency labeling scheme.
Spinrad~\cite[Section~6]{Spinrad95} also asked whether
strongly chordal and chordal bipartite graphs admit
$O(n\log^2 n)$-bit representations supporting efficient
adjacency queries and $O(\log^2 n)$-bit adjacency labeling schemes.
Dahlhaus~\cite{Dahlhaus97} studied representations
of chordal and strongly chordal graphs and proposed an
additional data structure supporting polylogarithmic-time
adjacency queries.
Otachi, Okamoto, and Yamazaki~\cite{OtachiOkamotoYamazaki07}
showed that every $P_6$-free chordal bipartite graph
is a unit grid intersection graph and discussed
the connection between grid intersection representations
and adjacency labeling for chordal bipartite graphs.
However, neither work resolved Spinrad's
adjacency-labeling question for the corresponding graph class.

\paragraph{Our results.}
Our main result is an asymptotically optimal labeling scheme
for $\Gamma$-free matrices, stated in the following theorem.

\begin{theorem}\label{thm:main}
Every $p \times q$ $\Gamma$-free matrix admits row and column
labels of $O(\log^2 N)$ bits, where $N=p+q$, such that any matrix
entry can be determined from the corresponding two labels
in $O(1)$ time in the word RAM model.
The label size is asymptotically optimal in the worst case.
\end{theorem}

Our construction processes the matrix row by row,
updating a forest after each row.
The support of each nonempty row forms a root path
in the resulting forest.
However, an earlier row's support need not remain
a root path in the final forest.
The main challenge is to encode the forest history
within the row and column labels, using only
$O(\log^2 N)$ bits per label.

The labeling scheme also gives asymptotically optimal
$O(\log^2 n)$-bit adjacency labels for $n$-vertex
strongly chordal and chordal bipartite graphs,
supporting constant-time adjacency queries.
Notably, these natural graph classes have superfactorial
but subquadratic speed $2^{\Theta(n\log^2 n)}$,
for which matching labeling bounds do not follow
from counting alone.
Consequently, we answer both representation questions
posed by Spinrad~\cite{Spinrad95}.

\section{Overview of Techniques}\label{sec:overview}

In this section, we outline the main ideas behind our labeling scheme.
We first describe how a $\Gamma$-free matrix can be represented by
a sequence of rooted forests, and then explain how a rank-based
path decomposition allows us to encode the forest history
using short labels.

\subsection{Forest representation of $\Gamma$-free matrices}
\label{sec:intro-structure}

We process the rows of a $\Gamma$-free matrix from top to bottom,
maintaining a rooted forest in which each column corresponds to a node.
For a rooted forest $F$, let $\Anc_F(x)$ and $\Desc_F(x)$ denote
the set of ancestors and descendants of a node $x$, respectively,
both including $x$.
Two nodes are said to be \textit{comparable} if one is an ancestor of the other.
We consider ordered forests, in which each parent
has a larger index than its children.
For two forests $F$ and $F'$ on the same nodes, we write
$F \preceq F'$ if every ancestor relation in $F$ is preserved in $F'$.

Let $F_i$ denote the forest after processing the first $i$ rows,
where $F_0$ consists of isolated nodes.
For each row $i$, we define its \emph{support} as
$S_i=\{x:M(i,x)=1\}$.
For each nonempty row $i$, let $s_i=\min S_i$.
We aim to construct the forests so that
\[
S_i = \Anc_{F_i}(s_i),
\]
and they form an \textit{ancestor-monotone sequence}
\[
F_0 \preceq F_1 \preceq \cdots \preceq F_p.
\]
The following theorem shows that such a representation exists
if and only if the matrix is $\Gamma$-free.

\begin{theorem}\label{thm:intro-characterization}
A binary matrix is $\Gamma$-free if and only if it admits
an ancestor-monotone sequence of ordered forests satisfying
the above properties.
\end{theorem}

A brief proof sketch is as follows.
For each nonempty row $i$, we arrange the nodes of
$S_i$ and their ancestors in $F_{i-1}$ into a path
in increasing column order, leaving all other parent
pointers unchanged.
Since $M$ is $\Gamma$-free, these ancestors already
belong to $S_i$, and all previous ancestor relations
are preserved.
Conversely, if every nonempty row support is a complete
root path in an ancestor-monotone forest sequence,
the matrix is $\Gamma$-free.
The detailed proof is given in \Cref{sec:dynamic-forest}.

We illustrate the construction with the following $\Gamma$-free matrix:
\[
M =
\begin{pmatrix}
1 & 0 & 1 & 0 & 1 \\
0 & 1 & 0 & 1 & 0 \\
1 & 1 & 1 & 1 & 1 \\
0 & 1 & 1 & 1 & 1
\end{pmatrix}.
\]
The corresponding forest sequence is
\[
\begin{aligned}
F_0 &: (1),\ (2),\ (3),\ (4),\ (5),\\
F_1 &: (1 \to 3 \to 5),\ (2),\ (4),\\
F_2 &: (1 \to 3 \to 5),\ (2 \to 4),\\
F_3 = F_4 &: (1 \to 2 \to 3 \to 4 \to 5),
\end{aligned}
\]
where $x \to y$ means that $y$ is the parent of $x$.

The third row merges the two paths of $F_2$ into a single path
in increasing column order, while preserving their ancestor relations.
Since $S_4=\Anc_{F_3}(2)$, processing the fourth row does not
change the forest.
Thus, the forest sequence preserves all ancestor relations,
and each row support $S_i$ corresponds to a complete root path
in $F_i$, as stated in \Cref{thm:intro-characterization}.

Note that the support of the first row is no longer a complete
root path in $F_4$, showing that the final forest alone is
insufficient to represent all rows.

\subsection{Rank-based path decomposition}
\label{sec:intro-persistence}
To determine whether $M(i,x)=1$, we need to check
whether $x\in\Anc_{F_i}(s_i)$.
One might consider using persistent data
structures~\cite{DriscollSarnakSleatorTarjan89}
to retain the intermediate forests.
However, this requires shared access, which is
unavailable to a decoder using only two labels.
Directly storing parent histories is also too costly:
for example, when
$S_i=\{1\}\cup\{q-i+1,\ldots,q\}$ for $1\le i<q$,
node $1$ changes its parent $\Theta(q)$ times.
Instead, we use a rank-based path decomposition
to encode the forest sequence compactly.

For each node $x$ in $F_t$, we define its \emph{rank} as
$\lfloor\log |\Desc_{F_t}(x)|\rfloor$.\footnote{All logarithms
are base two throughout the paper.}
An edge is called \emph{solid} if its endpoints have the same rank,
and \emph{dashed} otherwise.
The solid edges form vertex-disjoint paths, called \emph{solid paths},
including singleton nodes.
For two forests $F\preceq F'$, a node is \emph{stable}
if its rank is the same in $F$ and $F'$.
We assign an identifier to each solid path, where distinct
solid paths in the same forest have distinct identifiers.
We will show that these identifiers can be assigned so that
a stable node retains its solid-path identifier across each update.

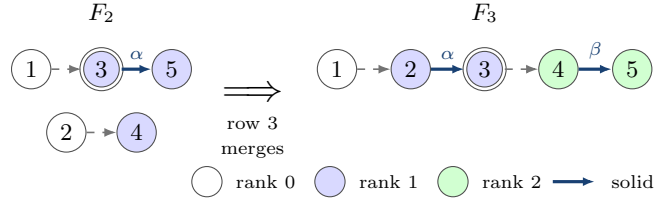
\begin{figure}[H]
\centering
\begin{tikzpicture}[
  x=0.93cm,y=0.93cm,
  vertex/.style={circle,draw=black!70,minimum size=5.2mm,inner sep=0pt,
                 font=\small},
  rankzero/.style={vertex,fill=white},
  rankone/.style={vertex,fill=blue!15},
  ranktwo/.style={vertex,fill=green!18},
  solid/.style={-{Latex[length=1.6mm]},very thick,draw=darkblue},
  dashedparent/.style={-{Latex[length=1.6mm]},semithick,dashed,draw=black!55}
]
  \node[font=\small\bfseries] at (1,1.15) {$F_2$};
  \node[rankzero] (f21) at (0,0.35) {$1$};
  \node[rankone,double,double distance=1pt] (f23) at (1,0.35) {$3$};
  \node[rankone] (f25) at (2,0.35) {$5$};
  \node[rankzero] (f22) at (0.5,-0.55) {$2$};
  \node[rankone] (f24) at (1.5,-0.55) {$4$};
  \draw[dashedparent] (f21) -- (f23);
  \draw[solid] (f23) -- node[above,font=\scriptsize,text=darkblue] {$\alpha$} (f25);
  \draw[dashedparent] (f22) -- (f24);

  \node[font=\Large] at (3.15,0) {$\Longrightarrow$};
  \node[font=\scriptsize,align=center] at (3.15,-0.62) {row $3$\\merges};

  \node[font=\small\bfseries] at (6.45,1.15) {$F_3$};
  \node[rankzero] (f31) at (4.35,0.35) {$1$};
  \node[rankone] (f32) at (5.4,0.35) {$2$};
  \node[rankone,double,double distance=1pt] (f33) at (6.45,0.35) {$3$};
  \node[ranktwo] (f34) at (7.5,0.35) {$4$};
  \node[ranktwo] (f35) at (8.55,0.35) {$5$};
  \draw[dashedparent] (f31) -- (f32);
  \draw[solid] (f32) -- node[above,font=\scriptsize,text=darkblue] {$\alpha$} (f33);
  \draw[dashedparent] (f33) -- (f34);
  \draw[solid] (f34) -- node[above,font=\scriptsize,text=darkblue] {$\beta$} (f35);

  \node[rankzero,minimum size=4mm] at (2.5,-1.25) {};
  \node[anchor=west,font=\scriptsize] at (2.78,-1.25) {rank $0$};
  \node[rankone,minimum size=4mm] at (4.25,-1.25) {};
  \node[anchor=west,font=\scriptsize] at (4.53,-1.25) {rank $1$};
  \node[ranktwo,minimum size=4mm] at (6.0,-1.25) {};
  \node[anchor=west,font=\scriptsize] at (6.28,-1.25) {rank $2$};
  \draw[solid] (7.42,-1.25) -- (8.02,-1.25);
  \node[anchor=west,font=\scriptsize] at (8.12,-1.25) {solid};
\end{tikzpicture}
\caption{Solid paths and their identifiers before and after
processing the third row of the example matrix in \Cref{sec:intro-structure}.}
\label{fig:running-merge}
\end{figure}

In \Cref{fig:running-merge}, nodes $3$ and $5$ form a solid path
of rank $1$ in $F_2$, with identifier $\alpha$.
After processing the third row, node $3$ remains at rank $1$,
and the new solid path consisting of nodes $2$ and $3$
inherits the identifier $\alpha$.
The solid path consisting of nodes $4$ and $5$ has rank $2$
and is assigned a new identifier $\beta$.

Although an update may change the solid paths, it does not change
whether two stable nodes belong to the same solid path.

\begin{theorem}[Rank stability]\label{thm:intro-stability}
Let $F\preceq F'$, and let $u,v$ be stable nodes
of the same rank $r$.
Then $u,v$ belong to the same solid path in $F$
if and only if they belong to the same solid path in $F'$.
\end{theorem}

We briefly outline the proof. Suppose two stable nodes
$u$ and $v$ of rank $r$ belong to different solid paths
in $F$ but the same solid path in $F'$.
Then they were incomparable in $F$, so their descendant
sets were disjoint, each containing at least $2^r$ nodes.
Since they become comparable in $F'$, one of them has
at least $2^{r+1}$ descendants, contradicting its stability.
The other direction follows from the preservation of ancestor
relations.
The formal proof is given in \Cref{sec:rankpaths}.

By \Cref{thm:intro-stability}, we can assign identifiers to solid paths
so that a stable node retains the identifier of its solid path.
A single update may change $\Theta(q)$ parent pointers,
but each node changes its rank and solid-path identifier only
$O(\log(q+1))$ times over the entire forest history.

Although solid-path decompositions are classical in dynamic-tree data structures~\cite{SleatorTarjan83},
\Cref{thm:intro-stability} establishes that nodes retaining their
ranks can also retain their solid-path identifiers across
ancestor-monotone updates.

\subsection{Construction of the labeling scheme}
We now describe how to construct the row and column labels
using the rank-based path decomposition.
For each $t\in\{0,\ldots,p\}$, we call $t$ the \emph{timestamp}
of the forest $F_t$.
For each column $x$, we store a list of triples $(t,r,h)$,
called \emph{records}, where $t$ is the timestamp, $r$ is
the rank of $x$, and $h$ is the identifier of the solid path
containing $x$ in $F_t$.
The first record is stored at timestamp $0$, and a new record
is added whenever the rank of $x$ or the identifier of its
solid path changes.
By \Cref{thm:intro-stability}, the identifier can change only
when the rank increases.
Since the rank increases at most $O(\log(q+1))$ times,
each column stores $O(\log(q+1))$ records.

For each nonempty row $i$, its support $S_i$ forms a complete
root path in $F_i$.
Since the ranks are nondecreasing along a root path,
$S_i$ intersects at most one solid path of each rank.
Thus, the label of row $i$ stores the identifier of each
intersected solid path and the smallest node of $S_i$
on that path.
Since there are $O(\log(q+1))$ possible ranks,
each row stores $O(\log(q+1))$ such pairs.

To answer $M(i,x)$, we first find the most recent record
of column $x$ whose timestamp is at most $i$.
Let $r$ and $h$ be the rank and solid-path identifier
stored in this record, respectively.
We then check whether the row label stores the same identifier $h$
for rank $r$, and whether $x$ is at least the smallest node
stored for that rank.
This determines whether $M(i,x)=1$, since the intersection
of a root path with a solid path is a suffix of the solid path,
whose nodes are ordered by their column indices.

Each column label stores $O(\log(q+1))$ records, and each row label
stores $O(\log(q+1))$ pairs of identifiers and boundaries.
Since each record or pair requires $O(\log N)$ bits,
the total label size is $O(\log^2 N)$ bits.
The predecessor search on the column records can be
performed in constant time in the word RAM model with
multiplication, using an additional data structure
stored within the column label without increasing its
asymptotic size.
Together with the forest representation in
\Cref{thm:intro-characterization}, this establishes
the upper bound of \Cref{thm:main}.
The detailed construction and proof are given in \Cref{sec:labels}.

Our forest representation is related to Spinrad's
prefix-reachability construction~\cite{Spinrad95}.
In particular, each forest $F_i$ is the transitive
reduction of the corresponding reachability relation
(see \Cref{app:spinrad}).
However, our labeling scheme uses persistent solid-path
identifiers to encode the forest history, allowing
each matrix entry to be determined using only two labels.

\paragraph{Organization.}
After the model in \Cref{sec:model},
\Cref{sec:dynamic-forest,sec:rankpaths,sec:labels,sec:optimality}
develop the forest representation, prove rank stability, construct the labels,
and establish optimality.
\Cref{sec:closure,sec:construction,sec:applications,sec:prior}
cover closure, construction time, applications, and related work.
\Cref{app:pseudocode} describes the encoder algorithm and the correspondence
with Spinrad's construction.
\Cref{app:complete-example} gives an example of our
labeling scheme, including the forest construction,
label assignment, query decoding, and closure computation.

\section{Model and Notation}
\label{sec:model}

Let $M$ be a $p\times q$ binary matrix, where $p,q\ge 1$,
and let $N=p+q$.
We index the rows and columns by $[p]$ and $[q]$, respectively,
where $[k]=\{1,\ldots,k\}$.
Unless otherwise stated, the row and column orders are fixed.

\begin{definition}\label{def:gamma-free}
A binary matrix $M$ is $\Gamma$-free if for every
$h<i$ and $a<b$,
\[
M(h,a)=M(h,b)=M(i,a)=1 \implies M(i,b)=1.
\]
\end{definition}

We next formally define the labeling scheme for binary matrices.

\begin{definition}\label{def:relation-labeling}
A \emph{matrix-labeling scheme} assigns a binary label
$\rho_M(i)$ to each row $i$ and a binary label
$\kappa_M(x)$ to each column $x$ of a matrix $M$.
A fixed decoder $D$, independent of $M$, determines
any matrix entry using only the corresponding two labels,
such that
\[
D(\rho_M(i),\kappa_M(x))=M(i,x).
\]
The decoder receives the row label first and the column
label second, and is not given the matrix dimensions or
any other information depending on $M$.
The label size is the maximum length of a row or column label.
\end{definition}

We use the word RAM model with word size
$w=c\lceil\log N\rceil$ for a sufficiently large constant $c$,
supporting standard word operations, including multiplication,
in constant time.
Each label is stored as an array of $w$-bit words, and the decoder
can access any word of the two given labels in constant time.

For a rooted forest $F$, let $\operatorname{Anc}_F(x)$ and
$\operatorname{Desc}_F(x)$ denote the sets of ancestors and
descendants of $x$, respectively, both including $x$.
Two nodes are comparable if one is an ancestor of the other.
For forests $F$ and $F'$ on the same node set, write
$F \preceq F'$ if every ancestor relation in $F$
is preserved in $F'$.
For each row $i$, let $S_i = \{x \in [q] : M(i,x) = 1\}$
and $s_i = \min S_i$ when $S_i \neq \emptyset$.
A rooted forest is \emph{ordered} if every parent has a
larger index than its children.
A set $U$ is \emph{upward closed} in $F$ if
$x\in U$ implies $\Anc_F(x)\subseteq U$.
A \emph{monotone forest representation} of $M$
consists of ordered forests
$F_0\preceq F_1\preceq\cdots\preceq F_p$,
where $F_0$ is edgeless and
$S_i=\Anc_{F_i}(s_i)$ for every nonempty row $i$.

\section{Forest Characterization of $\Gamma$-Free Matrices}
\label{sec:dynamic-forest}

We construct a sequence of forests by processing the rows of
$M$ in order, starting with the edgeless forest $F_0$.
For each row $i$, let $A_i$ consist of the nodes corresponding
to $S_i$ and all their ancestors in $F_{i-1}$:
\[
A_i=\bigcup_{x\in S_i}\Anc_{F_{i-1}}(x).
\]

If $S_i=\varnothing$, set $F_i=F_{i-1}$.
Otherwise, write $A_i = \{a_1 < \cdots < a_k\}$.
We construct $F_i$ by arranging $a_1,\ldots,a_k$
into a path in this order, with $a_k$ as its root.
All nodes outside $A_i$ retain their parent pointers.

\begin{lemma}\label{lem:upward-well-defined}
For every $i$, $A_i$ is upward closed in $F_{i-1}$,
and the update produces an ordered forest $F_i$.
\end{lemma}
\begin{proof}
Every node $x \in A_i$ is an ancestor of some
$s \in S_i$ in $F_{i-1}$.
Thus, all ancestors of $x$ also belong to $A_i$,
showing that $A_i$ is upward closed.
In $F_i$, every node has at most one parent,
and every parent has a larger index than its child.
Hence $F_i$ is an ordered forest.
\end{proof}

\begin{lemma}\label{lem:ancestor-permanence}
For every $i$, $F_{i-1}\preceq F_i$.
\end{lemma}
\begin{proof}
Consider any parent edge $x \to y$ in $F_{i-1}$.
If $x \notin A_i$, its parent pointer remains unchanged.
Otherwise, $y \in A_i$ by upward closure.
Since $x < y$, the new increasing path on $A_i$
preserves $y$ as an ancestor of $x$.
Thus, every parent edge in $F_{i-1}$ remains an
ancestor relation in $F_i$, implying
$F_{i-1} \preceq F_i$.
\end{proof}

We next show that $A_i\subseteq S_i$ when $M$ is $\Gamma$-free.
\begin{lemma}\label{lem:original-witness}
Suppose $M$ is $\Gamma$-free.
For every $i$, each parent edge $x\to y$ of $F_i$
has a witness row $h\le i$ such that
$M(h,x)=M(h,y)=1$.
Moreover, $A_i\subseteq S_i$.
\end{lemma}
\begin{proof}
We prove the claims by induction on $i$, starting with
the edgeless forest $F_0$.
For any $y\in A_i$, choose $s\in S_i$ such that
$y\in\Anc_{F_{i-1}}(s)$, and consider the path
$s=x_0<x_1<\cdots<x_t=y$ in $F_{i-1}$.
Since $s\in S_i$, we have $M(i,x_0)=1$.
For each parent edge $x_j\to x_{j+1}$, the induction
hypothesis gives a witness row $h<i$ with
$M(h,x_j)=M(h,x_{j+1})=1$.
By the $\Gamma$-free property, $M(i,x_j)=1$
implies $M(i,x_{j+1})=1$.
Thus $y\in S_i$, proving $A_i\subseteq S_i$.

Every newly installed parent edge $x\to y$ satisfies
$x,y\in A_i\subseteq S_i$, so
$M(i,x)=M(i,y)=1$. Every preserved parent edge has a witness row by
the induction hypothesis.
\end{proof}

\begin{corollary}\label{cor:direct-exactness}
If $M$ is $\Gamma$-free, then $A_i=S_i$ for every row $i$.
Moreover, $S_i=\Anc_{F_i}(s_i)$ for every nonempty row $i$.
\end{corollary}
\begin{proof}
Since every node is its own ancestor, $S_i\subseteq A_i$.
The preceding lemma gives $A_i\subseteq S_i$, so $A_i=S_i$.
For every nonempty row $i$, the update arranges $A_i$
into an increasing path whose maximum is a root.
Since $s_i=\min S_i=\min A_i$, it is the first node
of this path. Hence $S_i=\Anc_{F_i}(s_i)$.
\end{proof}

The forest sequence preserves ancestor relations and represents
each nonempty row as a root path.
We now prove \Cref{thm:intro-characterization}.

\begin{proof}[Proof of \Cref{thm:intro-characterization}]
If $M$ is $\Gamma$-free, the construction above gives
$F_0\preceq\cdots\preceq F_p$ by \Cref{lem:ancestor-permanence},
and $S_i=\Anc_{F_i}(s_i)$ for every nonempty row $i$
by \Cref{cor:direct-exactness}.

Conversely, suppose $M$ admits an ancestor-monotone sequence
of ordered forests $F_0\preceq\cdots\preceq F_p$,
where $F_0$ consists of isolated nodes and
$S_i=\Anc_{F_i}(s_i)$ for every nonempty row $i$
(i.e., $M$ has a monotone forest representation).
Let $h<i$ and $a<b$ satisfy
\[
a,b\in S_h,\qquad a\in S_i.
\]
Since $S_h=\Anc_{F_h}(s_h)$ and $F_h$ is ordered,
$b$ is an ancestor of $a$ in $F_h$.
Ancestor monotonicity preserves this relation in $F_i$.
Since $S_i=\Anc_{F_i}(s_i)$ contains $a$,
it also contains $b$.
Thus, $M$ is $\Gamma$-free.
\end{proof}

\section{Rank Stability and Persistent Identifiers}
\label{sec:rankpaths}
Consider an ancestor-monotone sequence of rooted forests
$F_0\preceq\cdots\preceq F_p$ on $q$ nodes.
For each node $x$ and time $t$, we use the notation
\[
\sizev_t(x)=|\Desc_{F_t}(x)|,
\qquad
\rankv_t(x)=\lfloor\log\sizev_t(x)\rfloor,
\qquad
L=\lfloor\log q\rfloor.
\]
As before, a parent edge is \emph{solid} if its endpoints
have the same rank, and \emph{dashed} otherwise.

The rank definition ensures that a subtree containing two
disjoint rank-$r$ subtrees has rank at least $r+1$.
We first show that ranks are nondecreasing along the
forest sequence.

\begin{lemma}\label{lem:rank-monotonicity}
For every node $x$, $\sizev_t(x)$ and $\rankv_t(x)$
are nondecreasing with $t$.
Consequently, $\rankv_t(x)$ changes at most $L$ times.
\end{lemma}
\begin{proof}
If $z\in\Desc_{F_{t-1}}(x)$, then
$x\in\Anc_{F_{t-1}}(z)$.
By ancestor monotonicity,
$x\in\Anc_{F_t}(z)$, so
$z\in\Desc_{F_t}(x)$.
Therefore, $\sizev_t(x)$ and $\rankv_t(x)$ are nondecreasing.
Since ranks lie in $\{0,\ldots,L\}$,
each rank changes at most $L$ times.
\end{proof}

We next consider the structure of solid edges.
A child of a node is called a \emph{solid child} if
the edge connecting it to its parent is solid.

\begin{lemma}\label{lem:solid-paths}
Every node in $F_t$ has at most one solid child.
Consequently, the solid edges partition $F_t$ into
vertex-disjoint paths, including singleton nodes.
Every root path intersects at most one solid path of each
rank and at most $L+1$ solid paths in total.
\end{lemma}
\begin{proof}
Suppose a node $z$ of rank $r$ has two distinct solid
children $x$ and $y$.
Their descendant sets are disjoint, and each contains
at least $2^r$ nodes. Therefore,
\[
\sizev_t(z)\ge 1+\sizev_t(x)+\sizev_t(y)
\ge 1+2^r+2^r>2^{r+1},
\]
contradicting $\rankv_t(z)=r$.
Hence the solid edges form vertex-disjoint paths.

Along any root path, subtree sizes and thus ranks are
nondecreasing. Therefore, nodes of the same rank form
a contiguous segment belonging to one solid path.
Since ranks lie in $\{0,\ldots,L\}$, every root path
intersects at most $L+1$ solid paths.
\end{proof}

We next characterize solid paths in terms of
comparability and rank.

\begin{lemma}\label{lem:solid-comparability}
Two nodes in $F_t$ lie on the same solid path if and only if
they are comparable and have the same rank.
\end{lemma}
\begin{proof}
Nodes on a solid path are comparable with equal ranks.
Conversely, if $u$ is a descendant of $v$ and both have
rank $r$, rank monotonicity along their path forces every
intermediate rank to equal $r$. All its edges are solid.
\end{proof}

\begin{figure}[H]
\centering
\begin{tikzpicture}[x=1cm,y=1cm]
\node[note] at (2.6,2.6) {\textbf{Before: disjoint descendant sets}};
\draw[draw=darkblue,fill=black!3] (1.2,1.95)--(.15,.25)--(2.25,.25)--cycle;
\draw[draw=darkblue,fill=black!3] (4,1.95)--(2.95,.25)--(5.05,.25)--cycle;
\node[vf,fill=white] at (1.2,1.95) {$u$};
\node[vf,fill=white] at (4,1.95) {$v$};
\node[note,font=\small] at (1.2,.7) {$U$};
\node[note,font=\small] at (4,.7) {$V$};
\node[note] at (2.6,-.2) {$|U|,|V|\ge2^r$,\quad $U\cap V=\varnothing$};
\node[font=\Large] at (6.05,1.15) {$\Longrightarrow$};
\node[note] at (9.65,2.6) {\textbf{After: suppose $v$ is above $u$}};
\draw[draw=darkblue,fill=black!3] (9.5,2)--(7.4,.05)--(11.6,.05)--cycle;
\node[vf,fill=white] at (9.5,2) {$v$};
\node[note] at (9.5,.65) {$U\cup V$};
\node[note] at (9.5,-.4) {$U\cup V\subseteq\Desc_{F'}(v)$\\$|\Desc_{F'}(v)|\ge2^{r+1}$: rank increases};
\end{tikzpicture}
\caption{The doubling argument when two previously incomparable
rank-$r$ nodes become comparable, where
$U=\Desc_F(u)$ and $V=\Desc_F(v)$.}
\label{fig:doubling-visual}
\end{figure}
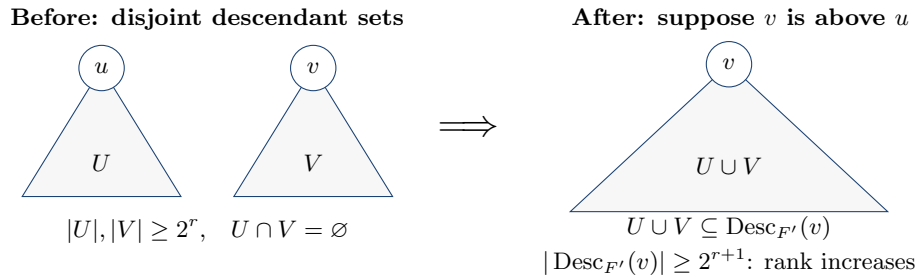

A node is \emph{stable} across $F \preceq F'$ if
its rank is unchanged.
We now prove \Cref{thm:intro-stability} using
the characterization in \Cref{lem:solid-comparability}.

\begin{proof}[Proof of \Cref{thm:intro-stability}]
If stable nodes $u,v$ lie on the same solid path in $F$,
they remain comparable in $F'$ by ancestor monotonicity.
Since their ranks remain $r$, \Cref{lem:solid-comparability}
implies that they lie on the same solid path in $F'$.

Conversely, suppose they lie on the same solid path in
$F'$ but not in $F$.
By \Cref{lem:solid-comparability}, they were incomparable
in $F$, so their descendant sets were disjoint.
Assume that $v$ is an ancestor of $u$ in $F'$.
By ancestor monotonicity, both old descendant sets
are contained in $\Desc_{F'}(v)$, as illustrated in
\Cref{fig:doubling-visual}. Thus,
\[
|\Desc_{F'}(v)|
\ge |\Desc_F(v)|+|\Desc_F(u)|
\ge 2^r+2^r=2^{r+1},
\]
contradicting the stability of $v$.
\end{proof}

We assign an identifier to each solid path, and each node
carries the identifier of the path containing it.
We now show that these identifiers can be maintained so
that stable nodes retain theirs across forest updates.

\begin{theorem}[Persistent solid-path identifiers]\label{thm:stable-ids}
The solid paths in $F_0,\ldots,F_p$ can be assigned
positive integer identifiers such that distinct paths
in each forest have distinct identifiers, and every
stable node retains its identifier across a transition.
In total, $O(q\log(q+1))$ distinct identifiers are used.
\end{theorem}

\begin{proof}
Initially assign distinct identifiers to the solid paths
in $F_0$. Consider a transition from $F_{t-1}$ to $F_t$.
If a new solid path $P'$ contains a stable node,
then by \Cref{thm:intro-stability}, all its stable nodes
belonged to the same old solid path $P$.
Assign the identifier of $P$ to $P'$.
Moreover, all stable nodes of an old solid path remain
on the same new solid path, so no two new paths inherit
the same identifier.

If a new solid path contains no stable node, assign it
a globally fresh identifier. Charge this identifier
to any node on the path, whose rank must have increased.
By \Cref{lem:rank-monotonicity}, each node can be charged
at most $L$ times. Thus, at most $qL$ fresh identifiers
are created, in addition to the at most $q$ initial
identifiers.
\end{proof}

Let $\sid_t(x)$ denote the identifier of the solid path
containing $x$ in $F_t$. Assigning identifiers consecutively
ensures that each uses $O(\log(q+1))$ bits.
Define the state of $x$ at time $t$ by
\[
\sigma_t(x)=\bigl(\rankv_t(x),\sid_t(x)\bigr).
\]
For each node $x$, we define its \emph{change list} as
a list of triples $(t,\rankv_t(x),\sid_t(x))$,
containing a record at $t=0$ and at every $t\ge 1$
for which $\sigma_t(x)\ne\sigma_{t-1}(x)$.
The following corollary bounds the length of each
change list in terms of the node's final rank.

\begin{corollary}\label{cor:short-histories}
The change list of each node $x$ contains at most
$1+\rankv_p(x)\le L+1$ records, and the total number
of records is $O(q\log(q+1))$.
\end{corollary}

\begin{proof}
By \Cref{thm:stable-ids}, the state of a node changes
only when its rank increases.
Since ranks are initially zero and never decrease,
the state of $x$ changes at most $\rankv_p(x)$ times.
Including the initial record proves the first bound,
and summing over all $q$ nodes gives the second.
\end{proof}

\section{Row and Column Labels}
\label{sec:labels}
We now use the change lists from \Cref{cor:short-histories}
to construct row and column labels for a matrix admitting
a monotone forest representation.

\begin{theorem}\label{thm:intro-compiler}
Let $M$ be a $p\times q$ binary matrix admitting a
monotone forest representation, and let $N=p+q$.
Then $M$ admits row and column labels of
$O(\log^2(q+1)+\log N)$ bits such that any entry
$M(i,x)$ can be determined from the corresponding
two labels in $O(1)$ time on the word RAM.
\end{theorem}

We prove the theorem in the remainder of this section.
Although each column's change list contains
$O(\log(q+1))$ records, each timestamp may require
$\log(p+1)$ bits.
To avoid this dependence on $p$, let $\nu(t)$ denote
the \emph{version} at time $t$, defined as the number
of forest changes up to that time:
\[
\nu(t)=|\{j\in[1,t]:F_j\ne F_{j-1}\}|.
\]

\begin{lemma}\label{lem:versions}
For every $0\le t\le p$, we have
$\nu(t)\le\binom q2$.
Thus, each version can be stored in
$O(\log(q+1))$ bits.
\end{lemma}

\begin{proof}
A forest is uniquely determined by its ancestor relations,
since the parent of each nonroot node is its lowest proper
ancestor. Thus, every strict change in the forest adds
at least one new ancestor relation.
By ancestor monotonicity, these relations never disappear.
Since there are at most $\binom q2$ such relations,
we have $\nu(t)\le\binom q2$.
\end{proof}

For each nonempty row $i$, the support
$S_i=\Anc_{F_i}(s_i)$ is a root path in $F_i$.
By \Cref{lem:solid-paths}, it intersects at most one
solid path of each rank. For each rank $r$ encountered
on $S_i$, let $H_i[r]$ be the identifier of that solid path
and $B_i[r]$ the smallest column index in its intersection
with $S_i$. The row label stores the version $\nu(i)$
and the arrays $H_i,B_i$, indexed by $0,\ldots,L$.
For ranks absent from $S_i$, set
$H_i[r]=0$ and $B_i[r]=q+1$.
For an empty row, use these default values at every rank.

For each column $x$, replace every timestamp $t$
in its change list by $\nu(t)$. The resulting list is
\[
\mathcal H(x)=((v_0,r_0,h_0),\ldots,(v_\ell,r_\ell,h_\ell)),
\]
where $v_0=0$, the versions are strictly increasing,
and $(r_j,h_j)$ is the state of $x$ from version $v_j$
until the next record.
Since an unchanged forest does not change any state,
this replacement preserves the change list.
By \Cref{cor:short-histories}, $\ell+1\le L+1$.

\begin{theorem}[Exact rank-path membership test]\label{thm:membership-test}
For every row $i$ and column $x$, let
$(r,h)=\sigma_i(x)$. Then
\[
M(i,x)=1\quad\Longleftrightarrow\quad
H_i[r]=h\ \text{ and }\ x\ge B_i[r].
\]
\end{theorem}

\begin{proof}
If row $i$ is empty, all its path identifiers are zero,
whereas $h$ is positive, so both sides are false.
Otherwise, suppose $M(i,x)=1$, or equivalently $x\in S_i$.
The row label stores the identifier $h$ of the solid path
containing $x$ at rank $r$. Since $B_i[r]$ is the smallest
node of $S_i$ on this path, $x\ge B_i[r]$.

Conversely, if $H_i[r]=h$, then $x$ and $B_i[r]$
lie on the same solid path in $F_i$.
Since node numbers increase toward the root,
$x\ge B_i[r]$ places $x$ on the suffix of this path
contained in $S_i$. Hence $M(i,x)=1$.
\end{proof}

\Cref{fig:decoder-core} illustrates the membership test.

\begin{figure}[ht]
\centering
\begin{tikzpicture}[x=1cm,y=1cm]
\node[note,anchor=west] at (0,2.05) {\textbf{Column label:} predecessor of $\nu(i)$};
\draw[-{Latex[length=2mm]},draw=darkblue] (0,.55)--(6,.55);
\foreach \x/\v in {.5/v_0,2.5/v_j,5/v_{j+1}} {
 \draw[darkblue] (\x,.45)--(\x,.65);
 \node[font=\small,below] at (\x,.4) {$\v$};}
\node[draw=darkblue,fill=softblue,inner sep=4pt,font=\small] at (2.5,1.12) {$(r,h)$};
\draw[-{Latex[length=1.6mm]},darkgreen,thick] (3.75,1.35)--(3.75,.62);
\node[font=\small,above] at (3.75,1.35) {$\nu(i)$};
\node[note,anchor=west] at (7,2.05) {\textbf{Row label:} inspect rank $r$};
\node[draw=darkblue,fill=softblue,inner sep=5pt] at (9.5,.7)
 {$\begin{array}{c|ccc}&\cdots&r&\cdots\\\hline H_i&\cdots&h&\cdots\\B_i&\cdots&z&\cdots\end{array}$};
\node[note] at (6,-.65) {Accept exactly when $H_i[r]=h$ and $x\ge B_i[r]=z$.};
\end{tikzpicture}
\caption{The membership test for $M(i,x)$: the column's active
record determines which entries of the row label are inspected.}
\label{fig:decoder-core}
\end{figure}
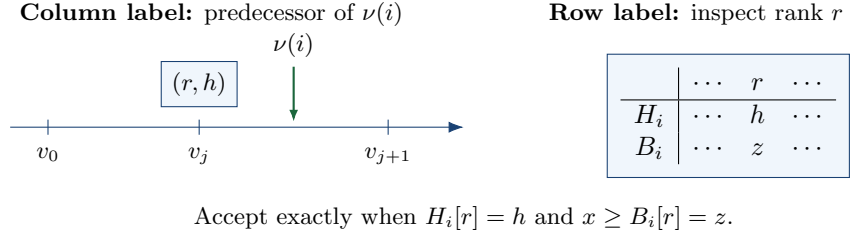

\subsection{Label Construction and Binary-Search Decoding}
\label{sec:elementary-decoder}

We store each field using $b=\Theta(\log(q+1))$ bits,
where $b$ is chosen as a sufficiently large power of two.
Since $q+1\le N$, the constant $c$ in
$w=c\lceil\log N\rceil$ from \Cref{sec:model}
can be chosen so that $b\le w$.

A row label for $i$ stores $q$, the version $\nu(i)$,
and the $L+1$ entries of each array $H_i$ and $B_i$.
A column label for $x$ stores $q$, the index $x$,
the number of records in $\mathcal H(x)$,
and three arrays containing their versions, ranks,
and identifiers.
These arrays are padded to $L+1$ entries.
Both labels also store the layout parameters $L$ and $b$,
which determine the array offsets. Since each column history contains an initial record
at version $0$, a predecessor exists for every
queried version.

\begin{lemma}\label{lem:label-size}
The fields of each row or column label require
$O(\log^2(q+1))$ bits in total.
\end{lemma}
\begin{proof}
A label has $O(L+1)$ fields of $b$ bits,
requiring $O((L+1)b)=O(\log^2(q+1))$ bits.
\end{proof}

By the preceding lemma, the fields occupy
$O(\log^2(q+1))$ bits.
We store $q$, $L$, $b$, and
$a=2^{\lfloor\log(w/b)\rfloor}$
in four fixed-width fields of a dedicated $w$-bit
header word, each occupying $w/4$ bits.
The remaining $b$-bit fields are packed into groups
of $a$ fields per $w$-bit word, so no field
crosses a word boundary.
Since $ab>w/2$, this packing uses
$O(\log^2(q+1)+w)$ bits and supports
constant-time access to any field using shifts and masks.
Thus, each label occupies
$O(\log^2(q+1)+\log N)$ bits.
We now describe a decoder that uses binary search
to find $\sigma_i(x)$ and applies
\Cref{thm:membership-test}.

\begin{algorithmdef}[Binary-search decoder]\label{alg:decoder}
Given a row label for $i$ and a column label for $x$:
\begin{enumerate}[leftmargin=2em]
\item Read the column index $x$ and the row's version $v=\nu(i)$
from their respective labels.
\item Binary-search the first $\ell+1$ entries of column $x$'s
version array to find $j=\max\{k:v_k\le v\}$,
where $\ell+1$ is the stored number of records.
\item Read $(r,h)=(r_j,h_j)$ from record $(v_j,r_j,h_j)$.
\item Return $1$ iff $H_i[r]=h$ and $x\ge B_i[r]$.
\end{enumerate}
\end{algorithmdef}

The binary-search decoder gives the following theorem.

\begin{theorem}\label{thm:binary-search-labels}
Every $p\times q$ binary matrix admitting a monotone
forest representation has row and column labels of
$O(\log^2(q+1)+\log N)$ bits supporting entry queries
in $O(1+\log\log(q+1))$ time on the word RAM,
where $N=p+q$.
\end{theorem}

\begin{proof}
The label-size bound follows from \Cref{lem:label-size} and
the packing above.
For a query $(i,x)$, we use binary search on the
change list of $x$ to find $\sigma_i(x)$.
The membership test of \Cref{thm:membership-test}
then correctly determines $M(i,x)$.
The list contains at most $L+1=O(\log(q+1))$
records, so binary search takes
$O(1+\log\log(q+1))$ time.
All other operations take constant time.
\end{proof}

\subsection{Constant-Time Decoding}
\label{sec:constant-time}

The binary-search decoder spends
$O(1+\log\log(q+1))$ time finding the state of a
column at the queried version.
We use the predecessor structure described in the
following theorem to replace binary search and
obtain constant-time decoding.

\begin{theorem}[Static fusion-tree predecessor~\cite{FredmanWillard93}]\label{thm:fusion}
On a word RAM with multiplication and $w$-bit words,
a static set of $k\ge 1$ distinct $w$-bit keys with
associated $w$-bit values can be stored in $O(k)$ words,
supporting predecessor queries in
$O(1+\log_w k)$ worst-case time.
\end{theorem}

We now complete the proof of \Cref{thm:intro-compiler}.

\begin{proof}[Proof of \Cref{thm:intro-compiler}]
For each column $x$, we build a predecessor structure
on the version numbers in its change list $\mathcal H(x)$.
By \Cref{cor:short-histories}, the number of records
satisfies $k\le L+1=O(b)$.
The version numbers and their associated rank-identifier
pairs fit in $b$ bits, where $b=\Theta(\log(q+1))$
is chosen sufficiently large.

We apply \Cref{thm:fusion} with word size $b$.
The resulting predecessor structure occupies $O(k)$
$b$-bit words and supports predecessor queries in
\[
O(1+\log_b k)=O(1)
\]
time. We store the structure, including its auxiliary
information, entirely in the column label.
Its space usage is
$O(kb)=O(\log^2(q+1))$ bits.

Since $b\le w$, operations on $b$-bit words,
including multiplication, can be simulated in
constant time on our word RAM.
Packing these words as described in
\Cref{sec:elementary-decoder}
preserves the $O(kb+w)$-bit space bound.
Thus, the total label size remains
$O(\log^2(q+1)+\log N)$ bits.

Replacing binary search in Algorithm~\ref{alg:decoder}
with predecessor search retrieves $\sigma_i(x)$
in constant time.
The remaining steps take constant time and are
correct by \Cref{thm:membership-test}.
\end{proof}

\Cref{app:complete-example} gives a complete example
of the forest construction, label assignment, and decoding.

\section{Information-Theoretic Optimality}
\label{sec:optimality}

By \Cref{thm:intro-characterization,thm:intro-compiler},
every $\Gamma$-free matrix admits
$O(\log^2 N)$-bit labels with constant-time decoding.
We now prove the matching lower bound and obtain
an optimal centralized representation.

\begin{theorem}[Spinrad~\cite{Spinrad95}]
\label{thm:spinrad}
Let $\gamma(m)$ denote the number of $m\times m$
$\Gamma$-free binary matrices with fixed row and
column orders. Then
\[
\gamma(m)=2^{\Theta(m\log^2 m)}.
\]
\end{theorem}

\begin{lemma}\label{lem:label-counting}
Let $\calF_{p,q}$ be a family of binary matrices
on the fixed row and column sets $[p]$ and $[q]$,
where $p+q=N$.
If a fixed decoder represents every matrix in
$\calF_{p,q}$ using labels of at most $b$ bits,
then
\[
b\ge \frac{\log|\calF_{p,q}|}{N}-1.
\]
\end{lemma}

\begin{proof}
There are fewer than $2^{b+1}$ binary strings
of length at most $b$, so there are fewer than
$2^{N(b+1)}$ possible assignments of labels
to the $N$ objects.
Since the decoder is fixed, each assignment
represents at most one matrix.
Taking logarithms gives the bound.
\end{proof}

\begin{theorem}\label{thm:optimal}
For $N\ge 2$, the worst-case maximum label length
over all $\Gamma$-free matrices with $p+q=N$
is $\Theta(\log^2 N)$.
\end{theorem}

\begin{proof}
The upper bound follows from
\Cref{thm:intro-characterization,thm:intro-compiler}.
For $N=2m$, \Cref{thm:spinrad} gives
$2^{\Omega(N\log^2 N)}$ distinct $\Gamma$-free
matrices.
By \Cref{lem:label-counting}, some label
must have length $\Omega(\log^2 N)$.
For odd $N$, the bound follows by adding
an all-zero row.
\end{proof}
This completes the proof of \Cref{thm:main}.
For fixed dimensions $p,q$, let $s=\min\{p,q\}$.
Padding an $s\times s$ $\Gamma$-free matrix with
zero rows or columns preserves $\Gamma$-freeness.
Since the labels of the original $2s$ rows and
columns determine the $s\times s$ matrix,
the counting argument also gives an
$\Omega(\log^2 s)$ lower bound on the
worst-case maximum label length.

A \textit{centralized representation} of $M$ is a data structure
that supports entry queries $M(i,x)$ given the row
and column indices $i$ and $x$.
Such a representation is obtained from
\Cref{thm:intro-compiler} by storing all row and
column labels together.

\begin{corollary}[Optimal centralized representation]
\label{cor:centralized}
For $N=p+q$, every $p\times q$ $\Gamma$-free matrix
admits an $O(N\log^2 N)$-bit centralized representation
supporting entry queries in $O(1)$ time on the word RAM.
This space bound is asymptotically optimal
in the worst case as a function of $N$.
\end{corollary}

\begin{proof}
For the centralized representation, we simply store
the $N$ row and column labels from
\Cref{thm:intro-characterization,thm:intro-compiler} in fixed-size blocks.
For a query $(i,x)$, we read the corresponding
two labels and run the decoder.
The total space is $O(N\log^2 N)$ bits,
and each query takes constant time.

By \Cref{thm:spinrad}, there are
$2^{\Omega(N\log^2 N)}$ distinct $\Gamma$-free
matrices.
A centralized representation must distinguish
these matrices and therefore requires
$\Omega(N\log^2 N)$ bits in the worst case.
\end{proof}

\section{Computing the $\Gamma$-Free Closure}
\label{sec:closure}

By \Cref{lem:upward-well-defined,lem:ancestor-permanence},
the forest construction in \Cref{sec:dynamic-forest}
produces an ancestor-monotone sequence of ordered forests
even when the input matrix is not $\Gamma$-free.
For an arbitrary $p\times q$ binary matrix $M^0$, let
$S_i=\{x:M^0(i,x)=1\}$ and apply the same construction,
obtaining forests $F_0,\ldots,F_p$ and sets
\[
A_i=\bigcup_{x\in S_i}\Anc_{F_{i-1}}(x).
\]
Define the output matrix $R_F$ by
\[
R_F(i,x)=1 \iff x\in A_i.
\]
For two matrices $P$ and $Q$ of the same dimensions,
write $P\le Q$ if $P(i,x)\le Q(i,x)$ for every entry $(i,x)$.
Since $S_i\subseteq A_i$, we have $M^0\le R_F$,
but equality need not hold.

Let $\gcl(M^0)$ denote the $\Gamma$-free closure
of $M^0$, defined as its entrywise least $\Gamma$-free
supermatrix. Such a matrix exists because the all-one
matrix is $\Gamma$-free and $\Gamma$-free matrices
are closed under entrywise intersection.

To show that $R_F=\gcl(M^0)$, we first prove that
$R_F$ is contained in every $\Gamma$-free supermatrix
of $M^0$.

\begin{lemma}\label{lem:closure-witness}
Let $M^0$ be an arbitrary $p\times q$ binary matrix,
and let $Q\ge M^0$ be any $\Gamma$-free matrix.
For every $i$, each parent edge $x\to y$ of $F_i$
has a witness row $h\le i$ such that
$Q(h,x)=Q(h,y)=1$.
Moreover,
\[
A_i\subseteq\{x:Q(i,x)=1\}.
\]
\end{lemma}

\begin{proof}
We extend the induction in \Cref{lem:original-witness}
to an arbitrary input matrix $M^0$ and a
$\Gamma$-free supermatrix $Q\ge M^0$.
We prove both claims by induction on $i$,
starting with the edgeless forest $F_0$.

For any $y\in A_i$, choose $s\in S_i$ such that
$y\in\Anc_{F_{i-1}}(s)$, and consider the path
$s=x_0<x_1<\cdots<x_t=y$ in $F_{i-1}$.
Since $s\in S_i$ and $Q\ge M^0$, we have
$Q(i,x_0)=1$.
For each parent edge $x_j\to x_{j+1}$,
the induction hypothesis gives a witness row
$h<i$ with $Q(h,x_j)=Q(h,x_{j+1})=1$.
By the $\Gamma$-free property of $Q$,
$Q(i,x_j)=1$ implies $Q(i,x_{j+1})=1$.
Thus $Q(i,y)=1$, proving the containment.

Every newly installed parent edge $x\to y$
satisfies $x,y\in A_i$, so
$Q(i,x)=Q(i,y)=1$.
Thus, row $i$ witnesses every newly installed edge,
while preserved parent edges retain their witnesses
by the induction hypothesis.
\end{proof}

By \Cref{lem:closure-witness}, every entry of $R_F$
is present in every $\Gamma$-free supermatrix of $M^0$.
We now show that $R_F$ itself is such a supermatrix.

\begin{theorem}[Canonical forest closure]
\label{thm:closure-forest}
For every $p\times q$ binary matrix $M^0$,
the forest construction produces
$R_F=\gcl(M^0)$ in $O(pq)$ time.
In particular, $M^0$ is $\Gamma$-free
if and only if $R_F=M^0$.
\end{theorem}

\begin{proof}
By \Cref{lem:upward-well-defined,lem:ancestor-permanence},
the forests form an ancestor-monotone sequence.
For each nonempty $A_i$, the construction arranges
its nodes into an increasing root path in $F_i$.
Thus, $R_F$ admits a monotone forest representation
and is $\Gamma$-free by \Cref{thm:intro-characterization}.

Since $S_i\subseteq A_i$ for every $i$,
we have $M^0\le R_F$.
Moreover, \Cref{lem:closure-witness} gives
$R_F\le Q$ for every $\Gamma$-free
supermatrix $Q$ of $M^0$.
Hence $R_F$ is the entrywise least
$\Gamma$-free supermatrix of $M^0$, so
$R_F=\gcl(M^0)$.

The equivalence follows from $R_F=\gcl(M^0)$
and the $\Gamma$-freeness of $R_F$.
The time bounds follow
from the construction analysis in
\Cref{sec:construction}.
\end{proof}

An example of the closure computation is given
in \Cref{app:closure-example}. By \Cref{thm:intro-compiler}, the closure also gives
$O(\log^2(q+1)+\log N)$-bit row and column labels
with constant-time decoding.

\section{Construction Time}
\label{sec:construction}
We analyze the time required to construct the canonical forest and the labels of \Cref{thm:intro-compiler}. The forest construction also computes the $\Gamma$-free closure of an arbitrary binary matrix, as established in \Cref{thm:closure-forest}. 

\begin{theorem}\label{thm:construction-time}
Given a $p\times q$ binary matrix $M^0$ with fixed
row and column orders, its $\Gamma$-free closure
and $O(\log^2(s+1)+\log N)$-bit row and column
labels can be constructed in $O(pq)$ time,
where $s=\min\{p,q\}$ and $N=p+q$.
The labels support entry queries in $O(1)$ time
on the word RAM.
\end{theorem}

\begin{proof}
Since $\Gamma$ is symmetric, taking the transpose
preserves $\Gamma$-freeness and commutes with
$\Gamma$-free closure. We may therefore orient
the input so that it has $r=\max\{p,q\}$ rows
and $s$ columns. We store one orientation bit in each label indicating whether
the input was transposed. If so, the decoder exchanges the
roles of the two input labels before applying the matrix decoder.

For each row $i$, let $S_i$ be its support.
To compute $A_i$, traverse the ancestor paths of nodes in $S_i$ in $F_{i-1}$, marking
visited nodes with the current row timestamp and stopping each traversal at the first
marked vertex. Each node is visited at most once during the row, so $A_i$ is found in
$O(s)$ time. 
Scanning the column order then installs $A_i$ as an increasing root path.

After each update, subtree sizes, ranks, and solid paths
can be recomputed in $O(s)$ time.
We assign identifiers as in \Cref{thm:stable-ids},
update column histories, and record the row snapshot.
Thus, the forest construction and label maintenance
take $O(rs)=O(pq)$ time.
Serializing the labels also takes $O(rs)$ time,
since each label contains $O(\log(s+1))$ fields.

To obtain constant-time decoding, we build a fusion
tree for the version
keys of each column history, as in \Cref{thm:fusion}.
We use $b=\Theta(\log(s+1))$ as the word size
for these structures, as in \Cref{sec:labels}.
By \Cref{cor:short-histories}, each history has
$k=O(b)$ keys, so the branching factor
$b^{\Theta(1)}$ gives a constant-height tree.
A fusion-tree node with $d$ keys can be constructed
deterministically in $O(d^4)$ time~\cite{FredmanWillard93}.
Hence, each predecessor structure can be constructed
in $b^{O(1)}$ time, including its auxiliary information.
The total construction time for all $s$ structures is
$O(s\log^{O(1)}(s+1))=O(s^2)=O(rs)$.
The label-size bound and constant-time decoding
follow from \Cref{thm:intro-compiler}.
\end{proof}

See \Cref{app:pseudocode} for the detailed construction algorithm.

\section{Applications}
\label{sec:applications}

We apply our labeling scheme to graph and hypergraph classes
whose associated matrices admit $\Gamma$-free orderings.
These orderings are needed to construct the labels,
but not to answer queries.
The corresponding centralized representations follow
by concatenating the labels, as in \Cref{cor:centralized}.

\subsection{Totally Balanced Matrices}

A \emph{cycle matrix} of order $k\ge3$ is a $k\times k$
matrix whose ones, after row and column permutations,
occur exactly at $(j,j)$ and $(j,j+1)$, with column
indices taken modulo $k$. A binary matrix is
\emph{totally balanced} if it contains no such submatrix.

\begin{theorem}[Anstee--Farber~\cite{AnsteeFarber84}]
\label{thm:AF}
A binary matrix is totally balanced if and only if
its rows and columns can be permuted independently
to make it $\Gamma$-free.
\end{theorem}

\begin{corollary}\label{cor:totally-balanced}
Every totally balanced $p\times q$ matrix admits
$O(\log^2 N)$-bit row and column labels supporting
entry queries in $O(1)$ time on the word RAM,
where $N=p+q$. The labels can be constructed
in $O(pq)$ time.
\end{corollary}

\begin{proof}
Permute the rows and columns as in \Cref{thm:AF}
and apply \Cref{thm:main} to the resulting
$\Gamma$-free matrix. Assign the resulting labels
to the corresponding rows and columns of the
original matrix.
A doubly lexical ordering of a $p\times q$ matrix
can be computed in $O(pq)$ time~\cite{Spinrad93},
and every such ordering of a totally balanced matrix
is $\Gamma$-free~\cite{Lubiw87}.
The construction time then follows from
\Cref{thm:construction-time}.
\end{proof}

\subsection{Chordal Bipartite Graphs}

For a bipartite graph $G=(U,V,E)$, its biadjacency matrix has
entry $(u,v)$ equal to one exactly when $uv\in E$.
The graph is \emph{chordal bipartite} if it has no induced
cycle of length at least six.

\begin{theorem}[Golumbic--Goss~\cite{GolumbicGoss78};
Anstee--Farber~\cite{AnsteeFarber84}]
\label{thm:CB-characterization}
A bipartite graph $G=(U,V,E)$ is chordal bipartite
if and only if the rows and columns of its biadjacency
matrix can be permuted independently to make it
$\Gamma$-free.
\end{theorem}

\begin{corollary}\label{cor:chordal-bipartite}
Every $n$-vertex chordal bipartite graph admits
$O(\log^2 n)$-bit adjacency labels supporting adjacency
queries in $O(1)$ time on the word RAM.
The labels can be constructed in $O(n^2)$ time.
This label-size bound is asymptotically optimal in the worst case.
\end{corollary}

\begin{proof}
Choose the vertex orders as in \Cref{thm:CB-characterization}
and apply \Cref{thm:main} to the biadjacency matrix.
For each vertex, we store its corresponding row or
column label and one bit indicating its side of
the bipartition.
If the queried vertices belong to the same side of
the bipartition, we return zero; otherwise, we use
the matrix decoder.
By \Cref{cor:totally-balanced} and \Cref{thm:construction-time},
the labels can be constructed in $O(n^2)$ time.

For the lower bound, fix a bipartition into two sets
of $m$ labeled vertices.
Every $m\times m$ $\Gamma$-free matrix defines a
chordal bipartite graph on these vertices, and
different matrices define different graphs.
By \Cref{thm:spinrad} and \Cref{lem:label-counting},
the worst-case label length is
$\Omega(\log^2 n)$.
\end{proof}

\subsection{Strongly Chordal Graphs and Digraphs}

A graph is \emph{chordal} if it has no induced cycle of length at least four.
A \emph{chord} joins two nonconsecutive cycle vertices. A graph is
\emph{strongly chordal} if it is chordal and every even cycle of length at least
six has a chord joining vertices at odd distance along that cycle~\cite{Farber83}.
For a simple graph $G$, let $M^+(G)$ be the closed-neighborhood matrix: its
$(u,v)$ entry is one exactly when $u=v$ or $uv\in E(G)$.

\begin{theorem}[Farber~\cite{Farber83}]
\label{thm:SC-characterization}
A simple graph $G$ is strongly chordal if and only if some vertex order,
used for both rows and columns, makes $M^+(G)$ $\Gamma$-free.
\end{theorem}

\begin{corollary}\label{cor:strongly-chordal}
Every $n$-vertex strongly chordal graph admits
$O(\log^2 n)$-bit adjacency labels supporting
adjacency queries in $O(1)$ time on the word RAM.
The labels can be constructed in $O(n^2)$ time.
This label-size bound is asymptotically optimal
in the worst case.
\end{corollary}

\begin{proof}
Order the vertices as in \Cref{thm:SC-characterization}
and apply \Cref{thm:main} to $M^+(G)$.
For each vertex, we store its row and column labels
and its index in the ordering.
To answer an adjacency query $(u,v)$, we first
compare their indices.
If they are equal, return zero; otherwise,
decode $M^+(G)(u,v)$ using the row label of $u$
and the column label of $v$.

The vertex ordering can be computed in
$O(n^2)$ time~\cite{Spinrad93}, and
\Cref{thm:construction-time} constructs the
labels in $O(n^2)$ time.

There are $2^{\Omega(n\log^2 n)}$ labeled
strongly chordal graphs~\cite{Spinrad95}.
By the counting argument of
\Cref{lem:label-counting}, some vertex label
must therefore use $\Omega(\log^2 n)$ bits.
\end{proof}

\paragraph{Strongly chordal split graphs.}
A graph is \emph{split} if its vertices can be partitioned
into a clique and an independent set.

\begin{corollary}\label{cor:split-boundary}
Every $n$-vertex strongly chordal split graph admits
$O(\log^2 n)$-bit adjacency labels supporting queries
in $O(1)$ time on the word RAM.
The labels can be constructed in $O(n^2)$ time.
This label-size bound is asymptotically optimal
in the worst case.
\end{corollary}

\begin{proof}
The upper bound and construction time follow from
\Cref{cor:strongly-chordal}.
For the lower bound, let $A$ be an $m\times m$
$\Gamma$-free matrix. Construct a split graph with
an independent set $Y$ and a clique $X$, each of
size $m$, where the edges between $X$ and $Y$
are specified by $A$. Ordering $Y$ before $X$,
its closed-neighborhood matrix is
\[
\begin{pmatrix}
I_m & A^\top\\
A & J_m
\end{pmatrix},
\]
where $I_m$ and $J_m$ are the identity and all-one
matrices, respectively.

Since $\Gamma$ is symmetric, $A^\top$ is also $\Gamma$-free.
We show that this matrix is $\Gamma$-free. Consider a candidate
forbidden pattern on rows $h<i$ and columns $a<b$.
If both rows belong to $Y$, the first selected column
must belong to $X$, so the pattern would lie in $A^\top$.
If both rows belong to $X$, either $b\in X$, making
the lower-right entry one, or both columns belong
to $Y$, where the pattern would lie in $A$.
If $h\in Y$ and $i\in X$, either $b\in X$, again
making the lower-right entry one, or both columns
belong to $Y$, where the upper row has at most one
nonzero entry. Thus, the constructed graph is
strongly chordal by \Cref{thm:SC-characterization}.
Since $A$ can be recovered from its cross edges,
\Cref{thm:spinrad} and \Cref{lem:label-counting} give
the matching $\Omega(\log^2 n)$ lower bound.
\end{proof}

\paragraph{Strongly chordal digraphs.}
Hell, Hern\'andez-Cruz, Huang, and Lin~
\cite{HellHernandezHuangLin19} define a digraph
(loops allowed) to be \emph{strongly chordal}
if some vertex order, used for both rows and
columns, makes its adjacency matrix $\Gamma$-free.

\begin{corollary}\label{cor:digraph}
Every $n$-vertex strongly chordal digraph admits
$O(\log^2 n)$-bit vertex labels supporting
arc queries in $O(1)$ time on the word RAM.
Given a $\Gamma$-free vertex ordering,
the labels can be constructed in $O(n^2)$ time.
This label-size bound is asymptotically optimal
in the worst case, even for loopless acyclic
two-layer digraphs.
\end{corollary}

\begin{proof}
For the upper bound, apply \Cref{thm:main}
to the adjacency matrix in a $\Gamma$-free order.
Each vertex stores both its row and column labels,
and a query $u\to v$ uses the row label of $u$
and the column label of $v$.
The construction time follows from
\Cref{thm:construction-time}.

For the lower bound, let $A$ be an $m\times m$
$\Gamma$-free matrix and consider
\[
D_A=\begin{pmatrix}
0&A\\
0&0
\end{pmatrix}.
\]
The matrix $D_A$ is $\Gamma$-free, since any
forbidden pattern would have to lie entirely
in $A$. It is the adjacency matrix of a
loopless acyclic two-layer digraph, and $A$
can be recovered from $D_A$.
By \Cref{thm:spinrad} and \Cref{lem:label-counting},
the matching $\Omega(\log^2 n)$ lower bound follows.
\end{proof}

\subsection{\texorpdfstring{$\beta$}{beta}-Acyclic Hypergraphs}

A finite hypergraph $\mathcal H=(V,\mathcal E)$ consists of
vertices $V$ and a family $\mathcal E$ of subsets of $V$.
Its incidence matrix has one row per hyperedge and
one column per vertex, with entry $(E,v)$ equal to one
exactly when $v\in E$.
A hypergraph is $\beta$-acyclic if and only if its incidence
matrix is totally balanced~\cite{Fagin83}.
A hypergraph is \emph{simple} if it has no repeated hyperedges.

\begin{corollary}\label{cor:beta}
Every $\beta$-acyclic hypergraph admits
$O(\log^2 N)$-bit incidence labels supporting
incidence queries in $O(1)$ time on the word RAM,
where $N=|V|+|\mathcal E|$.
The labels can be constructed in
$O(|V||\mathcal E|)$ time.
This label-size bound is asymptotically optimal
in the worst case, even for simple hypergraphs
with nonempty hyperedges.
\end{corollary}

\begin{proof}
The upper bound and construction time follow by applying
\Cref{cor:totally-balanced} to the incidence matrix.

For the lower bound, consider an $m\times m$
$\Gamma$-free matrix $A$ and construct
$B_A=(I_m\ A)$.
The matrix $B_A$ is $\Gamma$-free, since any
forbidden pattern would have to lie entirely in $A$.
Moreover, its rows are nonempty and pairwise distinct,
and $A$ can be recovered from $B_A$.
Thus, by \Cref{thm:spinrad}, there are
$2^{\Omega(m\log^2 m)}$ distinct simple
$\beta$-acyclic hypergraphs with $m$ hyperedges
and $2m$ vertices.
Since $N=3m$, \Cref{lem:label-counting} gives
the matching $\Omega(\log^2 N)$ lower bound.
\end{proof}

\section{Related Work}
\label{sec:prior}

\paragraph{$\Gamma$-free matrices and matrix orderings.}
Forbidden matrix patterns are used to characterize
strongly chordal graphs~\cite{Farber83},
totally balanced matrices~\cite{AnsteeFarber84},
chordal bipartite graphs~\cite{GolumbicGoss78},
and $\beta$-acyclic hypergraphs~\cite{Fagin83}.
Related work on matrix orderings includes
Lubiw's study of doubly lexical
orderings~\cite{Lubiw87} and the work of
Hell, Hern\'andez-Cruz, Huang, and
Lin~\cite{HellHernandezHuangLin19,HellHernandezHuangLin21}
on strong chordality in graphs and digraphs
through $\Gamma$-free orderings.

As discussed in \Cref{sec:introduction},
Spinrad~\cite{Spinrad95} established the counting bound
for $\Gamma$-free matrices and gave an asymptotically
optimal global representation. However, his representation
does not provide efficient access to individual matrix entries.
Dahlhaus~\cite{Dahlhaus97} studied compact representations
of chordal and strongly chordal graphs specified by
ordered graphs with $n$ vertices and $m$ edges, using
$O(n+m)$ words, and proposed an auxiliary data structure
supporting adjacency queries in polylogarithmic time. Otachi, Okamoto, and Yamazaki~\cite{OtachiOkamotoYamazaki07}
showed that every $P_6$-free chordal bipartite graph
is a unit grid intersection graph and asked whether
all chordal bipartite graphs are grid intersection graphs.
Chandran, Francis, and Mathew~\cite{ChandranFrancisMathew11}
answered this negatively by showing that chordal
bipartite graphs have unbounded boxicity.

\paragraph{Adjacency labeling and counting bounds.}
Kannan, Naor, and Rudich~\cite{KannanNaorRudich92}
studied implicit graph representations and their
connection to induced-universal graphs; see also
the survey of Gavoille and Peleg~\cite{GavoillePeleg03}.
An induced-universal graph for a class contains every
graph in the class as an induced subgraph.
Such a graph with $U$ vertices gives adjacency labels
of $\lceil\log U\rceil$ bits, although this does
not necessarily imply efficient decoding.

For general $n$-vertex graphs,
Alon~\cite{Alon17} constructed asymptotically optimal
induced-universal graphs with
$(1+o(1))2^{(n-1)/2}$ vertices, while
Alstrup, Kaplan, Thorup, and
Zwick~\cite{AlstrupKaplanThorupZwick19}
obtained $n/2+O(1)$-bit adjacency labels.
For restricted graph classes, stronger bounds are known.
Alstrup, Dahlgaard, and
Knudsen~\cite{AlstrupDahlgaardKnudsen17}
obtained $\log n+O(1)$-bit adjacency labels for forests,
equivalently giving induced-universal graphs
with $O(n)$ vertices.
Alon and Nenadov~\cite{AlonNenadov19}
constructed induced-universal graphs with
$O(n^{\Delta/2})$ vertices for graphs of maximum
degree at most any fixed $\Delta\ge 2$.

The number of labeled graphs in a hereditary class
as a function of $n$, called its \emph{speed},
has been studied by Scheinerman and Zito~\cite{ScheinermanZito94}
and Balogh, Bollob\'as, and Weinreich~\cite{BaloghBollobasWeinreich00}.
Atminas et al.~\cite{AtminasCollinsLozinZamaraev15}
investigated implicit representations of graph classes
with factorial speed.
For general hereditary classes, Bonamy, Esperet,
Groenland, and Scott~\cite{BonamyEsperetGroenlandScott21}
showed that every hereditary graph class $\mathcal{C}$
admits labels of $\mu_{\mathcal{C}}(n)+o(n)$ bits,
where $\mu_{\mathcal{C}}(n)=n^{-1}\log|\mathcal{C}_n|$
and $\mathcal{C}_n$ denotes the set of its
$n$-vertex graphs up to isomorphism.
For strongly chordal and chordal bipartite graphs,
$\mu_{\mathcal{C}}(n)=\Theta(\log^2 n)$.
However, the additional $o(n)$ term prevents
this bound from establishing $O(\log^2 n)$-bit labels.

As discussed in \Cref{sec:introduction}, counting bounds alone
do not always give matching labeling schemes.
Hatami and Hatami~\cite{HatamiHatami22}
constructed hereditary classes with
$2^{O(n\log n)}$ labeled graphs that require
$n^{\Omega(1)}$-bit labels.
Bonnet, Duron, Sylvester, Zamaraev, and
Zhukovskii~\cite{BonnetDuronSylvesterZamaraevZhukovskii24}
established a logarithmic separation between counting
bounds and label sizes for certain monotone graph classes.
For hereditary classes with at most $n!c^n$
labeled $n$-vertex graphs, Bonnet, Duron, Sylvester,
and Zamaraev~\cite{BonnetDuronSylvesterZamaraev25}
obtained polylogarithmic adjacency labels.

\paragraph{Structural approaches to adjacency labeling.}
Gavoille and Labourel~\cite{GavoilleLabourel07}
gave $(1+o(1))\log n$-bit labels for graphs of
bounded treewidth.
For planar graphs, Bonamy, Gavoille, and
Pilipczuk~\cite{BonamyGavoillePilipczuk22}
obtained $(4/3+o(1))\log n$-bit labels using
graph product structure, subsequently improved
to $(1+o(1))\log n$ bits~\cite{DujmovicEtAl21}.
The latter bound has recently been extended to
all fixed proper minor-closed graph
classes~\cite{DujmovicEtAl26}.
Strongly chordal graphs contain all cliques,
while chordal bipartite graphs contain all
complete bipartite graphs.
Thus, neither class has bounded treewidth or
belongs to a fixed proper minor-closed class.

\paragraph{Ancestry labeling, dynamic trees, and persistence.}
Ancestry labeling for static rooted trees has been
studied extensively~\cite{AbiteboulAlstrupKaplanMiloRauhe06}.
Fraigniaud and Korman~\cite{FraigniaudKorman16}
obtained $\log n+O(\log\log n)$-bit ancestry labels
with constant-time decoding.
Korman, Peleg, and Rodeh~\cite{KormanPelegRodeh04}
studied labeling in dynamic trees, while
Cohen, Kaplan, and Milo~\cite{CohenKaplanMilo10}
considered online ancestry labeling with immutable labels.

Solid-path decompositions are classical in dynamic-tree
data structures~\cite{SleatorTarjan83}, while persistence
allows access to earlier versions of shared data
structures~\cite{DriscollSarnakSleatorTarjan89}.

\section{Conclusion and Further Directions}
\label{sec:conclusion}

We have shown that $\Gamma$-free matrices admit
asymptotically optimal row and column labels supporting
constant-time entry queries. The main ingredient is
a monotone forest representation whose historical
paths can be encoded using persistent path identities.

\paragraph{Succinct representations and navigation.}
An open question is whether strongly chordal and chordal
bipartite graphs admit succinct representations supporting
efficient queries beyond adjacency. For either graph class,
let $H(n)$ denote the logarithm of the number of labeled
$n$-vertex graphs. Can adjacency be supported in
$H(n)+o(H(n))$ bits with constant query time?
Another question is whether an $O(n\log^2 n)$-bit
representation can support neighbor reporting in
$O(1+\deg(v))$ time, constant-time distance queries,
or shortest-path reporting in $O(1+\ell)$ time,
where $\ell$ is the number of edges in the reported path.

\paragraph{Multiple historical paths.}
Our forest characterization suggests a parameter
for binary matrices based on the number of
historical root paths needed to represent each row.
For fixed row and column orders, let $k$ be the
minimum integer for which an ancestor-monotone
sequence of ordered forests represents every row $i$
as a union of at most $k$ complete root paths in $F_i$.
By \Cref{thm:intro-characterization}, $k\le 1$
exactly for $\Gamma$-free matrices.
Minimizing $k$ over suitable row and column orders
also gives parameters for the corresponding graph
classes. It would be interesting to characterize
matrices and graphs with small $k$, determine
the complexity of computing these parameters,
and investigate labeling bounds in terms of $k$.

More generally, it would be interesting to identify
other evolving structures in which local identities
change only a small number of times, allowing
compact labels for historical queries.

\section*{AI Methodology}
We used ChatGPT Plus (OpenAI, GPT-6) to assist with
improving the presentation of the manuscript, preparing illustrative figures
and examples, and checking proofs for potential errors.
We also used Aristotle (Harmonic) for additional proof
checking and Claude (Anthropic) for feedback on
manuscript drafts.
All research ideas, main techniques, and mathematical
proofs were developed by the authors, who also wrote
the initial manuscript.
The authors independently verified the mathematical
results and references and take full responsibility
for the content of this paper.

\bibliographystyle{alpha}
\bibliography{ref}

\newpage
\appendix

\appendix

\section{Supplementary Construction Details}
\label{app:pseudocode}

\subsection{Encoder Algorithm}

The following algorithm constructs the labels of a
$\Gamma$-free matrix with fixed row and column orders,
using the canonical forest updates of
\Cref{sec:dynamic-forest} and the persistent
identifiers of \Cref{thm:stable-ids}.

\begin{algorithmdef}[Canonical $\Gamma$-free encoder]
\label{alg:encoder}
Input: a $\Gamma$-free matrix
$M\in\bits^{p\times q}$ with fixed row and column orders.

\begin{enumerate}[leftmargin=2em]
\item Initialize $F_0$ as the edgeless forest and assign
distinct positive identifiers to its singleton solid paths.
Set $\nu(0)=0$ and, for each column $x$, initialize
$H(x)=((0,0,\mathrm{id}_0(x)))$.

\item For $i=1,\ldots,p$, do the following:
  \begin{enumerate}[label=(\alph*),leftmargin=2em]
  \item Let $S_i=\{x:M(i,x)=1\}$ and compute
  \[
  A_i=\bigcup_{x\in S_i}\Anc_{F_{i-1}}(x).
  \]

  \item If $A_i\ne\varnothing$, replace it by the
  increasing root path on $A_i$, preserving every
  parent edge whose child lies outside $A_i$.
  Otherwise, set $F_i=F_{i-1}$.

  \item Set
  \[
  \nu(i)=\nu(i-1)+\mathbf 1[F_i\ne F_{i-1}].
  \]
  Compute the subtree sizes, ranks, and solid paths
  of $F_i$.

  \item Assign identifiers to the solid paths of $F_i$.
  A path containing a stable node inherits the
  identifier of its old solid path, as guaranteed by
  \Cref{thm:stable-ids}. Assign a fresh identifier
  to every remaining path.

  \item For each column $x$, append
  \[
  (\nu(i),\rankv_i(x),\sid_i(x))
  \]
  to $\mathcal H(x)$ if
  $\sigma_i(x)\ne\sigma_{i-1}(x)$.

  \item Initialize $H_i[r]=0$ and $B_i[r]=q+1$
  for every $r\in\{0,\ldots,L\}$.
  If $S_i\ne\varnothing$, traverse
  $\Anc_{F_i}(s_i)$ from $s_i=\min S_i$
  toward the root. At the first node $x$ of
  each encountered rank $r$, set
  \[
  H_i[r]=\sid_i(x),\qquad B_i[r]=x.
  \]
  Store $\nu(i)$ with the row snapshot.
  \end{enumerate}

\item Serialize the row snapshots and column histories
using the label format of \Cref{sec:labels}.
For constant-time decoding, construct the label-local
predecessor structure of \Cref{thm:fusion} on the
version numbers of each column history.
\end{enumerate}
\end{algorithmdef}

By \Cref{cor:direct-exactness}, each row support is the
root path used in Step~2(f).
The state records and row snapshots therefore satisfy
\Cref{thm:membership-test}.
The construction time is analyzed in
\Cref{thm:construction-time}.

\subsection{Relationship to Spinrad's Construction}
\label{app:spinrad}

Spinrad~\cite{Spinrad95} considers a directed graph
on the columns induced by their co-occurrence in row
prefixes. For each $i\in\{0,\ldots,p\}$, let $G_i$
be the digraph on $[q]$ containing an arc $x\to y$
if $x<y$ and some row $h\le i$ contains both
$x$ and $y$.

\begin{proposition}\label{prop:spinrad-correspondence}
For every $\Gamma$-free matrix $M$ and every prefix
$i$, the reflexive ancestor relation of $F_i$
equals the reflexive-transitive closure of $G_i$.
Equivalently, $F_i$ is the transitive reduction
of the reachability relation of $G_i$.
\end{proposition}

\begin{proof}
By \Cref{lem:original-witness}, every parent edge
$x\to y$ of $F_i$ has a witness row $h\le i$
containing both $x$ and $y$. Hence every ancestor
relation in $F_i$ is a reachability relation in $G_i$.

Conversely, suppose $x\to y$ is an arc of $G_i$,
witnessed by a row $h\le i$.
By \Cref{cor:direct-exactness}, the support of
row $h$ is a complete increasing root path in
$F_h$. Since $x<y$, the node $y$ is an ancestor
of $x$ in $F_h$. By
\Cref{lem:ancestor-permanence}, this relation
also holds in $F_i$.

Thus, the two reflexive-transitive relations
coincide. Since a forest is the transitive
reduction of its ancestor relation, the
proposition follows.
\end{proof}

\section{An Example of the Labeling Scheme}
\label{app:complete-example}

We illustrate the forest construction, the labels, and the
entry decoder on one matrix, then compute the $\Gamma$-free
closure of a modified matrix. The displayed records use
the version numbers of \Cref{lem:versions}; we omit
fixed-format padding and the predecessor structures of
\Cref{sec:constant-time}.

\subsection{The matrix and its forest history}

Consider a matrix with $p=9$ rows and $q=8$ columns,
with row supports $S_i$ given by
\[
\begin{array}{c|cccccccc|l}
 &1&2&3&4&5&6&7&8& S_i\\ \hline
1&1&0&1&0&0&0&1&0&\{1,3,7\}\\
2&0&1&0&1&0&0&0&1&\{2,4,8\}\\
3&0&0&0&0&1&1&0&0&\{5,6\}\\
4&1&1&1&1&0&0&1&1&\{1,2,3,4,7,8\}\\
5&0&0&1&1&0&0&1&1&\{3,4,7,8\}\\
6&1&1&1&1&1&1&1&1&\{1,2,3,4,5,6,7,8\}\\
7&0&0&0&1&1&1&1&1&\{4,5,6,7,8\}\\
8&0&0&0&0&0&0&0&0&\varnothing\\
9&0&1&1&1&1&1&1&1&\{2,3,4,5,6,7,8\}
\end{array}
\tag{B.1}\label{eq:complete-matrix}
\]
In the following table, $x\to y$ means that $y$ is
the parent of $x$; singleton components are omitted.
The canonical updates give
\[
\begin{array}{c|l|l}
t&A_t=S_t\ (t\ge1)&\text{nontrivial components of }F_t\\ \hline
0&\text{--}&\text{none}\\
1&\{1,3,7\}&1\to3\to7\\
2&\{2,4,8\}&1\to3\to7,\quad2\to4\to8\\
3&\{5,6\}&1\to3\to7,\quad2\to4\to8,\quad5\to6\\
4&\{1,2,3,4,7,8\}&1\to2\to3\to4\to7\to8,\quad5\to6\\
5&\{3,4,7,8\}&F_5=F_4\\
6&\{1,2,3,4,5,6,7,8\}&1\to2\to3\to4\to5\to6\to7\to8\\
7&\{4,5,6,7,8\}&F_7=F_6\\
8&\varnothing&F_8=F_7\\
9&\{2,3,4,5,6,7,8\}&F_9=F_8.
\end{array}
\tag{B.2}\label{eq:complete-forests}
\]
Rows~$1$, $2$, and $3$ create three components.
Row~$4$ merges the first two, and row~$6$ merges the
remaining component into the same path. The supports
of rows~$5$, $7$, and $9$ are already complete root
paths, so these rows do not change the forest; row~$8$
is empty. For every $t=1,\ldots,9$, the pre-update
ancestor union is $A_t=S_t$. The distinct forests
are shown in \Cref{fig:history-visual}. By
\Cref{thm:intro-characterization}, the matrix is
$\Gamma$-free.

\begin{figure}[!htb]
\centering
\begin{tikzpicture}[
x=1.22cm,y=0.80cm,
ve/.append style={shorten >=0.5pt,shorten <=0.5pt},
vd/.append style={dash pattern=on 1.5pt off 1pt}]
\node[note,anchor=west] at (-0.3,0.8) {\textbf{$F_0$}};
\node[vf,fill=white] (h0x1) at (0.0,0.0) {$1$};
\node[vf,fill=white] (h0x2) at (0.75,0.0) {$2$};
\node[vf,fill=white] (h0x3) at (1.5,0.0) {$3$};
\node[vf,fill=white] (h0x4) at (2.25,0.0) {$4$};
\node[vf,fill=white] (h0x5) at (3.0,0.0) {$5$};
\node[vf,fill=white] (h0x6) at (3.75,0.0) {$6$};
\node[vf,fill=white] (h0x7) at (4.5,0.0) {$7$};
\node[vf,fill=white] (h0x8) at (5.25,0.0) {$8$};
\node[note,anchor=west] at (6.7,0.8) {\textbf{$F_1$}};
\node[vf,fill=white] (h1x1) at (7.0,0.0) {$1$};
\node[vf,fill=blue!15] (h1x3) at (7.75,0.0) {$3$};
\node[vf,fill=blue!15] (h1x7) at (8.5,0.0) {$7$};
\draw[vd] (h1x1)--(h1x3);
\draw[ve] (h1x3)--(h1x7);
\node[vf,fill=white] (h1x2) at (7.0,-0.85) {$2$};
\node[vf,fill=white] (h1x4) at (7.75,-0.85) {$4$};
\node[vf,fill=white] (h1x5) at (8.5,-0.85) {$5$};
\node[vf,fill=white] (h1x6) at (9.25,-0.85) {$6$};
\node[vf,fill=white] (h1x8) at (10.0,-0.85) {$8$};

\node[note,anchor=west] at (-0.3,-2.7) {\textbf{$F_2$}};
\node[vf,fill=white] (h2x1) at (0.0,-3.5) {$1$};
\node[vf,fill=blue!15] (h2x3) at (0.75,-3.5) {$3$};
\node[vf,fill=blue!15] (h2x7) at (1.5,-3.5) {$7$};
\draw[vd] (h2x1)--(h2x3);
\draw[ve] (h2x3)--(h2x7);
\node[vf,fill=white] (h2x2) at (0.0,-4.35) {$2$};
\node[vf,fill=blue!15] (h2x4) at (0.75,-4.35) {$4$};
\node[vf,fill=blue!15] (h2x8) at (1.5,-4.35) {$8$};
\draw[vd] (h2x2)--(h2x4);
\draw[ve] (h2x4)--(h2x8);
\node[vf,fill=white] (h2x5) at (0.0,-5.2) {$5$};
\node[vf,fill=white] (h2x6) at (0.75,-5.2) {$6$};
\node[note,anchor=west] at (6.7,-2.7) {\textbf{$F_3$}};
\node[vf,fill=white] (t3x1) at (7.0,-3.5) {$1$};
\node[vf,fill=blue!15] (t3x3) at (7.75,-3.5) {$3$};
\node[vf,fill=blue!15] (t3x7) at (8.5,-3.5) {$7$};
\draw[vd] (t3x1)--(t3x3);
\draw[ve] (t3x3)--(t3x7);
\node[vf,fill=white] (t3x2) at (7.0,-4.35) {$2$};
\node[vf,fill=blue!15] (t3x4) at (7.75,-4.35) {$4$};
\node[vf,fill=blue!15] (t3x8) at (8.5,-4.35) {$8$};
\draw[vd] (t3x2)--(t3x4);
\draw[ve] (t3x4)--(t3x8);
\node[vf,fill=white] (t3x5) at (7.0,-5.2) {$5$};
\node[vf,fill=blue!15] (t3x6) at (7.75,-5.2) {$6$};
\draw[vd] (t3x5)--(t3x6);

\node[note,anchor=west] at (-0.3,-6.2) {\textbf{$F_4$}};
\node[vf,fill=white] (t4x1) at (0.0,-7.0) {$1$};
\node[vf,fill=blue!15] (t4x2) at (0.75,-7.0) {$2$};
\node[vf,fill=blue!15] (t4x3) at (1.5,-7.0) {$3$};
\node[vf,fill=green!18] (t4x4) at (2.25,-7.0) {$4$};
\node[vf,fill=green!18] (t4x7) at (3.0,-7.0) {$7$};
\node[vf,fill=green!18] (t4x8) at (3.75,-7.0) {$8$};
\draw[vd] (t4x1)--(t4x2);
\draw[ve] (t4x2)--(t4x3);
\draw[vd] (t4x3)--(t4x4);
\draw[ve] (t4x4)--(t4x7);
\draw[ve] (t4x7)--(t4x8);
\node[vf,fill=white] (t4x5) at (0.0,-7.85) {$5$};
\node[vf,fill=blue!15] (t4x6) at (0.75,-7.85) {$6$};
\draw[vd] (t4x5)--(t4x6);

\node[note,anchor=west] at (6.7,-6.2) {\textbf{$F_6$}};
\node[vf,fill=white] (t6x1) at (7.0,-7.0) {$1$};
\node[vf,fill=blue!15] (t6x2) at (7.75,-7.0) {$2$};
\node[vf,fill=blue!15] (t6x3) at (8.5,-7.0) {$3$};
\node[vf,fill=green!18] (t6x4) at (9.25,-7.0) {$4$};
\node[vf,fill=green!18] (t6x5) at (10.0,-7.0) {$5$};
\node[vf,fill=green!18] (t6x6) at (10.75,-7.0) {$6$};
\node[vf,fill=green!18] (t6x7) at (11.5,-7.0) {$7$};
\node[vf,fill=orange!25] (t6x8) at (12.25,-7.0) {$8$};
\draw[vd] (t6x1)--(t6x2);
\draw[ve] (t6x2)--(t6x3);
\draw[vd] (t6x3)--(t6x4);
\draw[ve] (t6x4)--(t6x5);
\draw[ve] (t6x5)--(t6x6);
\draw[ve] (t6x6)--(t6x7);
\draw[vd] (t6x7)--(t6x8);
\node[note] at (6,-10)
{$F_5=F_4$;\qquad $F_7=F_8=F_9=F_6$.\\
Rank colours: white $0$, blue $1$, green $2$, orange $3$.};
\end{tikzpicture}
\caption{The forest history of \eqref{eq:complete-matrix}.
All eight vertices are shown in every panel.
Solid arrows join equal ranks; dashed arrows cross ranks.
Rows $1,2,3$ create components, rows $4,6$ merge them,
and the remaining rows leave the forest unchanged.
The row supports and snapshots are given in
\eqref{eq:complete-forests} and \eqref{eq:complete-snapshots}.}
\label{fig:history-visual}
\end{figure}

\subsection{Ranks and path identifiers}

Here $L=\lfloor\log 8\rfloor=3$. Initially, the eight
singleton paths have identifiers $1,\ldots,8$.
At later times, identifiers follow the inheritance
rule of \Cref{thm:stable-ids}. The ranks and
identifiers, listed in column order, are
\[
\begin{array}{c|c|c}
t&(\rankv_t(1),\ldots,\rankv_t(8))
 &(\sid_t(1),\ldots,\sid_t(8))\\ \hline
0&(0,0,0,0,0,0,0,0)&(1,2,3,4,5,6,7,8)\\
1&(0,0,1,0,0,0,1,0)&(1,2,9,4,5,6,9,8)\\
2&(0,0,1,1,0,0,1,1)&(1,2,9,10,5,6,9,10)\\
3&(0,0,1,1,0,1,1,1)&(1,2,9,10,5,11,9,10)\\
4&(0,1,1,2,0,1,2,2)&(1,9,9,12,5,11,12,12)\\
5&(0,1,1,2,0,1,2,2)&(1,9,9,12,5,11,12,12)\\
6&(0,1,1,2,2,2,2,3)&(1,9,9,12,12,12,12,13)\\
7,8,9&(0,1,1,2,2,2,2,3)&(1,9,9,12,12,12,12,13).
\end{array}
\tag{B.3}\label{eq:complete-states}
\]
Identifiers $9$, $10$, and $11$ first appear at times
$1$, $2$, and $3$, respectively, and identifiers $12$
and $13$ at times $4$ and $6$. An identifier changes
only when the node's rank increases. The merges in
\Cref{fig:merge-detail} also show that nodes retaining
their ranks keep their identifiers when their solid
paths change.

The forest changes only at times $1,2,3,4,6$,
so the stored version numbers are
\[
\begin{array}{c|rrrrrrrrrr}
t&0&1&2&3&4&5&6&7&8&9\\\hline
\nu(t)&0&1&2&3&4&4&5&5&5&5
\end{array}
\tag{B.4}\label{eq:example-versions}
\]

The column histories, with unchanged states omitted,
are
\[
\begin{array}{c|l@{\qquad}c|l}
x&\mathcal H(x)&x&\mathcal H(x)\\ \hline
1&(0,0,1)&5&(0,0,5),(5,2,12)\\
2&(0,0,2),(4,1,9)&6&(0,0,6),(3,1,11),(5,2,12)\\
3&(0,0,3),(1,1,9)&7&(0,0,7),(1,1,9),(4,2,12)\\
4&(0,0,4),(2,1,10),(4,2,12)&8&(0,0,8),(2,1,10),(4,2,12),(5,3,13).
\end{array}
\tag{B.5}\label{eq:complete-histories}
\]
Each triple records a version, rank, and identifier.
Column~$8$, for instance, changes at row times
$2,4,6$, which correspond to stored versions
$2,4,5$. Its four records attain the $L+1$ bound
of \Cref{cor:short-histories}.

\begin{figure}[!htb]
\centering
\begin{tikzpicture}[x=1cm,y=1cm]
\node[note,anchor=west] at (-.4,1) {\textbf{Time 3}};
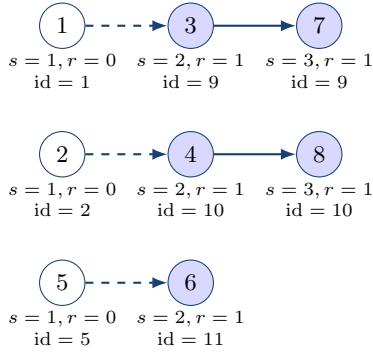
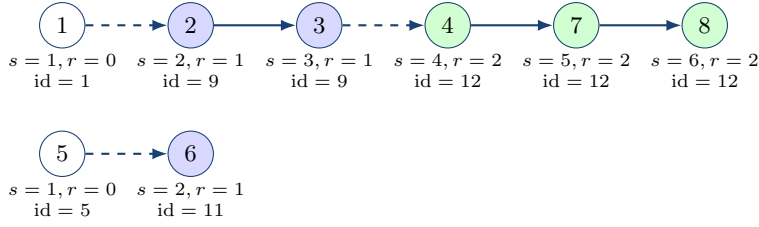
\node[vf,fill=white] (t3x1) at (0.0,0.0) {$1$};
\node[font=\scriptsize,align=center] at (0.0,-0.6)
{$s=1,r=0$\\$\mathrm{id}=1$};
\node[vf,fill=blue!15] (t3x3) at (1.7,0.0) {$3$};
\node[font=\scriptsize,align=center] at (1.7,-0.6)
{$s=2,r=1$\\$\mathrm{id}=9$};
\node[vf,fill=blue!15] (t3x7) at (3.4,0.0) {$7$};
\node[font=\scriptsize,align=center] at (3.4,-0.6)
{$s=3,r=1$\\$\mathrm{id}=9$};
\draw[vd] (t3x1)--(t3x3);
\draw[ve] (t3x3)--(t3x7);
\node[vf,fill=white] (t3x2) at (0.0,-1.7) {$2$};
\node[font=\scriptsize,align=center] at (0.0,-2.3)
{$s=1,r=0$\\$\mathrm{id}=2$};
\node[vf,fill=blue!15] (t3x4) at (1.7,-1.7) {$4$};
\node[font=\scriptsize,align=center] at (1.7,-2.3)
{$s=2,r=1$\\$\mathrm{id}=10$};
\node[vf,fill=blue!15] (t3x8) at (3.4,-1.7) {$8$};
\node[font=\scriptsize,align=center] at (3.4,-2.3)
{$s=3,r=1$\\$\mathrm{id}=10$};
\draw[vd] (t3x2)--(t3x4);
\draw[ve] (t3x4)--(t3x8);
\node[vf,fill=white] (t3x5) at (0.0,-3.4) {$5$};
\node[font=\scriptsize,align=center] at (0.0,-4.0)
{$s=1,r=0$\\$\mathrm{id}=5$};
\node[vf,fill=blue!15] (t3x6) at (1.7,-3.4) {$6$};
\node[font=\scriptsize,align=center] at (1.7,-4.0)
{$s=2,r=1$\\$\mathrm{id}=11$};
\draw[vd] (t3x5)--(t3x6);

\node[note,anchor=west] at (-.4,-5) {\textbf{Time 4}};
\node[vf,fill=white] (t4x1) at (0.0,-6) {$1$};
\node[font=\scriptsize,align=center] at (0.0,-6.6)
{$s=1,r=0$\\$\mathrm{id}=1$};
\node[vf,fill=blue!15] (t4x2) at (1.7,-6) {$2$};
\node[font=\scriptsize,align=center] at (1.7,-6.6)
{$s=2,r=1$\\$\mathrm{id}=9$};
\node[vf,fill=blue!15] (t4x3) at (3.4,-6) {$3$};
\node[font=\scriptsize,align=center] at (3.4,-6.6)
{$s=3,r=1$\\$\mathrm{id}=9$};
\node[vf,fill=green!18] (t4x4) at (5.1,-6) {$4$};
\node[font=\scriptsize,align=center] at (5.1,-6.6)
{$s=4,r=2$\\$\mathrm{id}=12$};
\node[vf,fill=green!18] (t4x7) at (6.8,-6) {$7$};
\node[font=\scriptsize,align=center] at (6.8,-6.6)
{$s=5,r=2$\\$\mathrm{id}=12$};
\node[vf,fill=green!18] (t4x8) at (8.5,-6) {$8$};
\node[font=\scriptsize,align=center] at (8.5,-6.6)
{$s=6,r=2$\\$\mathrm{id}=12$};
\draw[vd] (t4x1)--(t4x2);
\draw[ve] (t4x2)--(t4x3);
\draw[vd] (t4x3)--(t4x4);
\draw[ve] (t4x4)--(t4x7);
\draw[ve] (t4x7)--(t4x8);
\node[vf,fill=white] (t4x5) at (0.0,-7.7) {$5$};
\node[font=\scriptsize,align=center] at (0.0,-8.3)
{$s=1,r=0$\\$\mathrm{id}=5$};
\node[vf,fill=blue!15] (t4x6) at (1.7,-7.7) {$6$};
\node[font=\scriptsize,align=center] at (1.7,-8.3)
{$s=2,r=1$\\$\mathrm{id}=11$};
\draw[vd] (t4x5)--(t4x6);

\node[note,anchor=west] at (-.4,-9.5) {\textbf{Time 6}};
\node[vf,fill=white] (t6x1) at (0.0,-10.5) {$1$};
\node[font=\scriptsize,align=center] at (0.0,-11.1)
{$s=1,r=0$\\$\mathrm{id}=1$};
\node[vf,fill=blue!15] (t6x2) at (1.7,-10.5) {$2$};
\node[font=\scriptsize,align=center] at (1.7,-11.1)
{$s=2,r=1$\\$\mathrm{id}=9$};
\node[vf,fill=blue!15] (t6x3) at (3.4,-10.5) {$3$};
\node[font=\scriptsize,align=center] at (3.4,-11.1)
{$s=3,r=1$\\$\mathrm{id}=9$};
\node[vf,fill=green!18] (t6x4) at (5.1,-10.5) {$4$};
\node[font=\scriptsize,align=center] at (5.1,-11.1)
{$s=4,r=2$\\$\mathrm{id}=12$};
\node[vf,fill=green!18] (t6x5) at (6.8,-10.5) {$5$};
\node[font=\scriptsize,align=center] at (6.8,-11.1)
{$s=5,r=2$\\$\mathrm{id}=12$};
\node[vf,fill=green!18] (t6x6) at (8.5,-10.5) {$6$};
\node[font=\scriptsize,align=center] at (8.5,-11.1)
{$s=6,r=2$\\$\mathrm{id}=12$};
\node[vf,fill=green!18] (t6x7) at (10.2,-10.5) {$7$};
\node[font=\scriptsize,align=center] at (10.2,-11.1)
{$s=7,r=2$\\$\mathrm{id}=12$};
\node[vf,fill=orange!25] (t6x8) at (11.9,-10.5) {$8$};
\node[font=\scriptsize,align=center] at (11.9,-11.1)
{$s=8,r=3$\\$\mathrm{id}=13$};
\draw[vd] (t6x1)--(t6x2);
\draw[ve] (t6x2)--(t6x3);
\draw[vd] (t6x3)--(t6x4);
\draw[ve] (t6x4)--(t6x5);
\draw[ve] (t6x5)--(t6x6);
\draw[ve] (t6x6)--(t6x7);
\draw[vd] (t6x7)--(t6x8);
\end{tikzpicture}
\caption{The two merges $F_3\to F_4$ and
$F_5=F_4\to F_6$. Each vertex carries
its subtree size $s$, rank $r=\lfloor\log s\rfloor$,
and identifier. Vertex $3$ carries identifier $9$
into $2\to3$; vertices $4,7$ carry identifier $12$
into $4\to5\to6\to7$. Vertex $8$ reaches
rank $3$ and receives identifier $13$.}
\label{fig:merge-detail}
\end{figure}
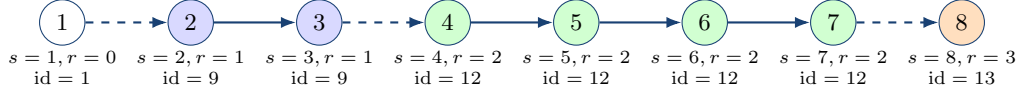

\FloatBarrier
\subsection{Row snapshots}

For absent ranks, set $H_i[r]=0$ and $B_i[r]=q+1=9$.
The row paths give the following arrays, indexed
by ranks $0,1,2,3$:
\[
\begin{array}{c|c|c}
i&(H_i[0],H_i[1],H_i[2],H_i[3])
 &(B_i[0],B_i[1],B_i[2],B_i[3])\\ \hline
1&(1,9,0,0)&(1,3,9,9)\\
2&(2,10,0,0)&(2,4,9,9)\\
3&(5,11,0,0)&(5,6,9,9)\\
4&(1,9,12,0)&(1,2,4,9)\\
5&(0,9,12,0)&(9,3,4,9)\\
6&(1,9,12,13)&(1,2,4,8)\\
7&(0,0,12,13)&(9,9,4,8)\\
8&(0,0,0,0)&(9,9,9,9)\\
9&(0,9,12,13)&(9,2,4,8).
\end{array}
\tag{B.6}\label{eq:complete-snapshots}
\]
For example, row~$5$ has path $P_5=3\to4\to7\to8$.
Its rank-$1$ part begins at node~$3$ of the
identifier-$9$ path $2\to3$, while its rank-$2$
part is the entire identifier-$12$ path $4\to7\to8$.
Thus $H_5[1]=9$, $B_5[1]=3$, $H_5[2]=12$,
and $B_5[2]=4$; the other ranks use sentinels.
The boundary matters here: column~$2$ has the
same identifier~$9$ but is not in row~$5$.

\subsection{Updates and stored records}
\label{sec:example-ledger}

At time~$0$, initialize $\mathcal H(x)=((0,0,x))$
for every $x\in[8]$. For each row~$i$, the table
lists the forest update, the column-history appends,
and the new row snapshot. The appended records use
version $\nu(i)$ from \eqref{eq:example-versions},
not the row index $i$.
For example, $x\in\{3,7\}:(1,1,9)$ means that
$(1,1,9)$ is appended to both column histories.
Unlisted histories remain unchanged, and absent
ranks use $(H_i[r],B_i[r])=(0,9)$.
Each row stores $\nu(i)$ and $H_i,B_i$; each column
stores its index and history, as described in
\Cref{sec:labels}.

\begin{table}[!htbp]
\centering
\small
\renewcommand{\arraystretch}{1.10}
\begin{tabularx}{\linewidth}
{@{}p{.035\linewidth}X p{.36\linewidth}@{}}
\toprule
$i$ & Forest update and column-history appends
& Row snapshot\\
\midrule
1 & Create $1\to3\to7$; other vertices stay isolated.\newline
Append $x\in\{3,7\}:(1,1,9)$.
& $H_1=(1,9,0,0)$\newline
$B_1=(1,3,9,9)$\\[3pt]
2 & Create $2\to4\to8$ beside the first path.\newline
Append $x\in\{4,8\}:(2,1,10)$.
& $H_2=(2,10,0,0)$\newline
$B_2=(2,4,9,9)$\\[3pt]
3 & Create $5\to6$ as the third component.\newline
Append $x=6:(3,1,11)$.
& $H_3=(5,11,0,0)$\newline
$B_3=(5,6,9,9)$\\[3pt]
4 & Merge the first two components into
$1\to2\to3\to4\to7\to8$; retain $5\to6$.\newline
Append $x=2:(4,1,9)$ and\newline
$x\in\{4,7,8\}:(4,2,12)$.
& $H_4=(1,9,12,0)$\newline
$B_4=(1,2,4,9)$\\[3pt]
\midrule
5 & Forest unchanged; select $3\to4\to7\to8$.\newline
No column-history append.
& $H_5=(0,9,12,0)$\newline
$B_5=(9,3,4,9)$\\[3pt]
6 & Merge into the full chain
$1\to2\to3\to4\to5\to6\to7\to8$.\newline
Append $x\in\{5,6\}:(5,2,12)$ and\newline
$x=8:(5,3,13)$.
& $H_6=(1,9,12,13)$\newline
$B_6=(1,2,4,8)$\\[3pt]
7 & Forest unchanged; select $4\to5\to6\to7\to8$.\newline
No column-history append.
& $H_7=(0,0,12,13)$\newline
$B_7=(9,9,4,8)$\\[3pt]
8 & Empty support; forest unchanged.\newline
No column-history append.
& $H_8=(0,0,0,0)$\newline
$B_8=(9,9,9,9)$\\[3pt]
9 & Forest unchanged; select
$2\to3\to4\to5\to6\to7\to8$.\newline
No column-history append.
& $H_9=(0,9,12,13)$\newline
$B_9=(9,2,4,8)$\\
\bottomrule
\end{tabularx}
\caption{Forest updates, column-history appends, and row
snapshots for all nine rows. Each listed record is
appended to the corresponding column history.}
\label{tab:ledger}
\label{tab:ledger-early}
\label{tab:ledger-late}
\end{table}

At the first merge, at time~$4$, column~$3$ keeps
rank~$1$ and identifier~$9$, although its solid path
changes from $3\to7$ to $2\to3$. It therefore needs
no new history record. Only columns~$2$, $4$, $7$,
and $8$ append records, as listed in
\Cref{tab:ledger}. The new row snapshot uses
$B_4[1]=2$.

At time~$5$, no column state changes, but row~$5$
begins at node~$3$ rather than~$2$ on the solid path
with identifier~$9$. Its boundary $B_5[1]=3$ excludes
column~$2$, even though the identifiers match.

At time~$6$, nodes~$4$ and $7$ retain rank~$2$ and
identifier~$12$ as nodes~$5$ and $6$ join their solid
path. Columns~$4$ and $7$ require no new records;
columns~$5$ and $6$ append their new rank-$2$ states.
Column~$5$ moves directly from rank~$0$ to rank~$2$,
without a rank-$1$ record. Column~$8$ moves to
rank~$3$ and appends its fourth record.

The histories in \eqref{eq:complete-histories}
contain eight initial records and twelve appends.

\clearpage
\subsection{Entry queries}

The following queries use the decoder of
\Cref{sec:labels}.
\Cref{fig:labels-visual} shows the row-$5$
label fields and three example queries.
For a query $(i,x)$, the triple $(v;r,h)$
denotes the predecessor record selected at
version $\nu(i)$ from \eqref{eq:example-versions}.

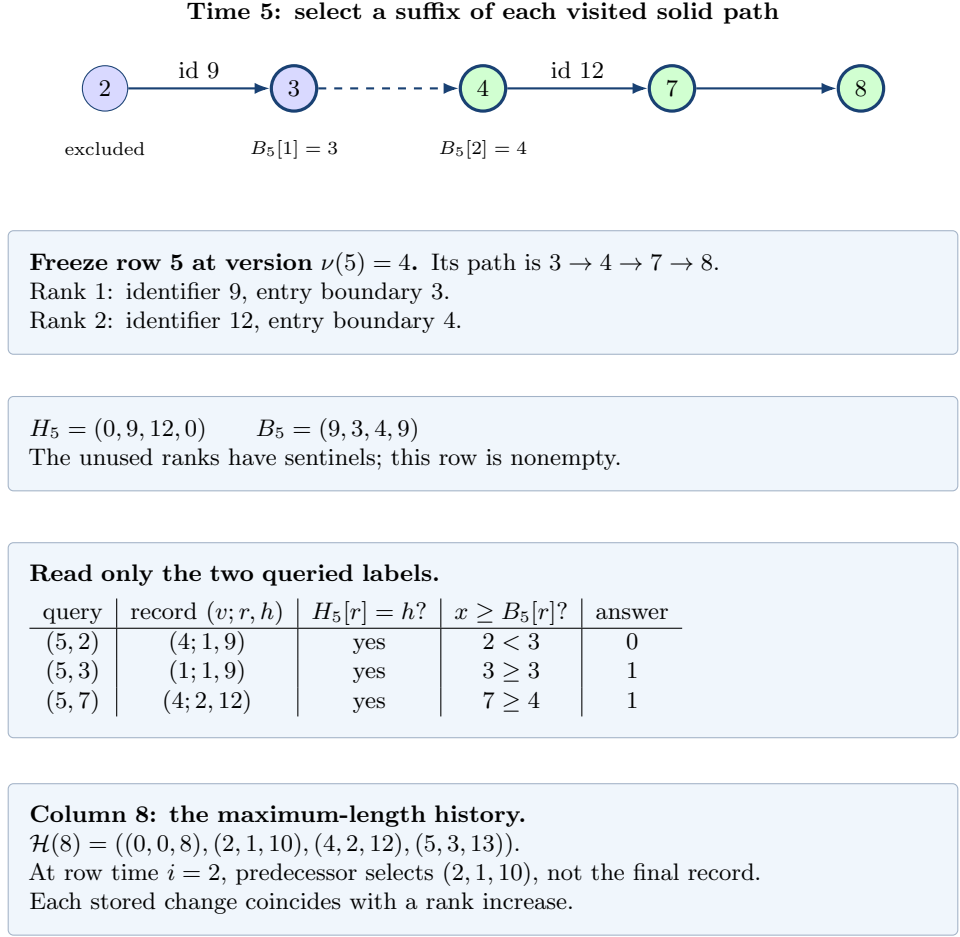
\begin{figure}[!htb]
\centering
\begin{tikzpicture}[x=1cm,y=1cm]
\node[note] at (6,6.7)
{\textbf{Time 5: select a suffix of each visited solid path}};
\node[vf,fill=blue!15] (l2) at (1,5.7) {$2$};
\node[vf,fill=blue!15,very thick] (l3) at (3.5,5.7) {$3$};
\node[vf,fill=green!18,very thick] (l4) at (6,5.7) {$4$};
\node[vf,fill=green!18,very thick] (l7) at (8.5,5.7) {$7$};
\node[vf,fill=green!18,very thick] (l8) at (11,5.7) {$8$};
\draw[ve](l2)--node[above,font=\small]{id 9}(l3);
\draw[vd](l3)--(l4);
\draw[ve](l4)--node[above,font=\small]{id 12}(l7);
\draw[ve](l7)--(l8);
\node[note,font=\scriptsize] at (1,4.9) {excluded};
\node[note,font=\scriptsize] at (3.5,4.9) {$B_5[1]=3$};
\node[note,font=\scriptsize] at (6,4.9) {$B_5[2]=4$};
\node[boxnote,text width=12cm] at (6,3)
{\textbf{Freeze row 5 at version $\nu(5)=4$.}
Its path is $3\to4\to7\to8$.\\
Rank 1: identifier 9, entry boundary 3.\\
Rank 2: identifier 12, entry boundary 4.};
\node[boxnote,text width=12cm] at (6,1)
{$H_5=(0,9,12,0)$\qquad $B_5=(9,3,4,9)$\\
The unused ranks have sentinels; this row is nonempty.};
\node[boxnote,text width=12cm] at (6,-1.6)
{\textbf{Read only the two queried labels.}\\[4pt]
$\begin{array}{c|c|c|c|c}
\text{query}&\text{record }(v;r,h)&H_5[r]=h?
&x\ge B_5[r]?&\text{answer}\\\hline
(5,2)&(4;1,9)&\text{yes}&2<3&0\\
(5,3)&(1;1,9)&\text{yes}&3\ge3&1\\
(5,7)&(4;2,12)&\text{yes}&7\ge4&1
\end{array}$};
\node[boxnote,text width=12cm] at (6,-4.5)
{\textbf{Column 8: the maximum-length history.}\\
$\mathcal H(8)=((0,0,8),(2,1,10),(4,2,12),(5,3,13))$.\\
At row time $i=2$, predecessor selects $(2,1,10)$,
not the final record.\\
Each stored change coincides with a rank increase.};
\end{tikzpicture}
\caption{From a historical path to label fields
and answers. The version selects the active column
record, the identifier selects the solid path, and
the boundary selects its relevant suffix.
See \eqref{eq:example-versions},
\eqref{eq:complete-histories}, and
\eqref{eq:complete-snapshots}.}
\label{fig:labels-visual}
\end{figure}

\begin{enumerate}[leftmargin=2.6em]
\item $(2,8)$ selects $(2;1,10)$. Here $H_2[1]=10$
and $8\ge B_2[1]=4$, giving $M(2,8)=1$.

\item $(2,3)$ selects $(1;1,9)$, but $H_2[1]=10\ne9$.
Thus $M(2,3)=0$.

\item $(5,2)$ selects $(4;1,9)$. The identifiers
agree, but $2<B_5[1]=3$, so $M(5,2)=0$.

\item $(5,3)$ selects $(1;1,9)$. Both
$H_5[1]=9$ and $3\ge B_5[1]=3$ hold,
so $M(5,3)=1$.

\item $(5,7)$ selects $(4;2,12)$, with
$H_5[2]=12$ and $7\ge B_5[2]=4$.
Thus $M(5,7)=1$.

\item $(5,5)$ selects $(0;0,5)$ and gives
$M(5,5)=0$, since $H_5[0]=0\ne5$.
Using the later record $(5;2,12)$ would
incorrectly give $1$: $H_5[2]=12$ and
$5\ge B_5[2]=4$. The predecessor search
must therefore exclude later versions.

\item $(7,1)$ selects $(0;0,1)$. Since
$H_7[0]=0$, we have $M(7,1)=0$.

\item All queries $(8,x)$ give $0$, because
$H_8$ contains only zero identifiers.

\item $(9,8)$ selects $(5;3,13)$. We have
$H_9[3]=13$ and $8\ge B_9[3]=8$,
so $M(9,8)=1$.
\end{enumerate}

Evaluating all $9\cdot8=72$ entries from
\Cref{eq:complete-histories,eq:complete-snapshots}
gives, row by row,
\[
10100010,\ 01010001,\ 00001100,\
11110011,\ 00110011,\ 11111111,\
00011111,\ 00000000,\ 01111111,
\]
in agreement with \eqref{eq:complete-matrix}.

\subsection{Closure after deleting an entry}
\label{app:closure-example}

Delete entry $(5,8)$ from \eqref{eq:complete-matrix}
and call the resulting matrix $M^-$. Rows~$4$, $5$
and columns~$3<8$ now contain the forbidden pattern
$\bigl(\begin{smallmatrix}1&1\\1&0\end{smallmatrix}\bigr)$.
Thus every $\Gamma$-free supermatrix of $M^-$
must restore $(5,8)$.

Before row~$5$ is processed, the forest contains
$1\to2\to3\to4\to7\to8$. Hence
\[
\bigcup_{x\in\{3,4,7\}}
\Anc_{F_4}(x)=\{3,4,7,8\}.
\]
The forest update adds only $(5,8)$, recovering
the $\Gamma$-free matrix \eqref{eq:complete-matrix}.
Because every $\Gamma$-free supermatrix of $M^-$
contains this entry, the recovered matrix is the
$\Gamma$-free closure of $M^-$, as in
\Cref{thm:closure-forest}.

\begin{figure}[!htb]
\centering
\begin{tikzpicture}[x=1cm,y=1cm]
\node[boxnote] at (2.5,1)
{$\begin{array}{c|cc}
&3&8\\\hline
4&1&1\\
5&1&\color{red}{0}
\end{array}$};
\node[note] at (8.5,1)
{Delete $(5,8)$ from the example.\\
The $\Gamma$ implication forces it to be one.};
\node[boxnote,text width=12cm] at (6,-1.3)
{\textbf{The same implication in the old forest:}\\
$3\to4\to7\to8$ lies in $F_4$.\\
Row 5 after deletion selects $\{3,4,7\}$;
its ancestor union is $\{3,4,7,8\}$.};
\node[note] at (6,-3)
{The restored matrix is $\Gamma$-free,
and every completion contains $(5,8)$.};
\end{tikzpicture}
\caption{The deleted entry $(5,8)$ is forced
both by the forbidden-pattern condition
and by the canonical forest update.}
\label{fig:closure-visual}
\end{figure}

\end{document}